\documentclass{LMCS}

\def\dOi{11(2:12)2015}
\lmcsheading%
{\dOi}
{1--37}
{}
{}
{Mar.~19, 2014}
{Jun.~23, 2015}
{}

\ACMCCS{[{\bf Theory of computation}]: Models of computation---Probabilistic computation}

\usepackage[utf8]{inputenc}
\usepackage{amsfonts,amssymb,amsmath,amsthm}

\usepackage[ruled]{algorithm2e}
\renewcommand{\algorithmcfname}{ALGORITHM}
\SetAlFnt{\small}
\SetAlCapFnt{\small}
\SetAlCapNameFnt{\small}
\SetAlCapHSkip{0pt}
\IncMargin{-\parindent}

\theoremstyle{plain}
\theoremstyle{plain}\newtheorem{problem}[thm]{Problem}
\theoremstyle{plain}
\theoremstyle{plain}
\theoremstyle{plain}
\theoremstyle{plain}

\newcommand{\pRob}[1]{\mathbb{P}_{#1}}
\renewcommand{\AA}{\mathcal{A}} 
\newcommand{\BB}{\mathcal{B}}
\newcommand{\CC}{\mathcal{C}}

\newcommand{\DD}{\mathcal{D}} 
\newcommand{\supp}{\mathrm{Supp}} 

\newcommand{\set}[1]{\{ #1 \}}
\newcommand{\val}[1]{\text{val}(#1)}

\newcommand{\Cl}{{\bf Cl}}
\newcommand{\rank}{\text{rank}}

\newcommand{\lima}{{\bf a}}
\newcommand{\limb}{{\bf b}}

\newcommand{\limz}{{\bf z}}

\newcommand{\limu}{{\bf u}}
\newcommand{\limv}{{\bf v}}
\newcommand{\limw}{{\bf w}}
\newcommand{\bolda}{{\bf a}}
\newcommand{\bolde}{{\bf e}}
\newcommand{\boldeps}{{\bf 1}}

\newcommand{\PSPACE}{\mathrm{PSPACE}}

\newcommand{\JJ}{\mathcal{J}}

\newcommand{\NN}{\mathbb{N}}
\newcommand{\nNN}{_{n\in\NN}}

\newcommand{\monoid}{\mathcal{G}}
\newcommand{\monoidext}{\mathcal{G}_+}
\newcommand{\MM}{\mathcal{M}}

\newcommand{\tendsto}[1]{\mathop{\longrightarrow}_{n}}

\newcommand{\pmin}{p_{\min}}
\newcommand{\mmaxdepth}{2^{3^{|Q|^2 + 1}}}
\newcommand{\dabound}{\pmin^{\mmaxdepth}}

\newcommand{\buchi}{\textrm{B\"uchi}}
\newcommand{\cobuchi}{\textrm{CoB\"uchi}}
\newcommand{\safe}{\textrm{Safe}}
\newcommand{\reach}{\textrm{Reach}}
\newcommand{\parc}{\textrm{Parity}(c)}
\newcommand{\parity}{\textrm{Parity}}

\begin{document}

\title[Probabilistic Leaktight Automata]{Deciding the Value 1 Problem \\ for Probabilistic Leaktight Automata\rsuper*}

\author[N.~Fijalkow]{Nathana\"el Fijalkow\rsuper a}	
\address{{\lsuper a}LIAFA, Universit{\'e} Denis Diderot - Paris~7, France,
and University of Warsaw, Poland.}	
\email{nath@liafa.univ-paris-diderot.fr}  

\author[H.~Gimbert]{Hugo Gimbert\rsuper b}	
\address{{\lsuper b}LaBRI, CNRS, Bordeaux, France.}	
\email{hugo.gimbert@labri.fr}  

\author[E.~Kelmendi]{Edon Kelmendi\rsuper c}
\address{{\lsuper c}LaBRI and Université de Bordeaux, France.}
\email{edon.kelmendi@labri.fr}

\author[Y.~Oualhadj]{Youssouf Oualhadj\rsuper e}	
\address{{\lsuper e}Universit\'e Paris-Est, LACL, France.}	
\email{youssouf.oualhadj@lacl.fr}  

\titlecomment{{\lsuper*}A preliminary version appeared in LiCS'2012~\cite{FGO12}.
The sections about probabilistic automata over infinite words and the comparisons with structurally simple automata are new.
The latter is mostly due to Edon Kelmendi.}

\thanks{This project was supported by the french ANR project "Stoch-MC" as well as "LaBEX CPU" of Université de Bordeaux.}


\begin{abstract}
The value 1 problem is a decision problem for probabilistic automata over finite words:
given a probabilistic automaton, are there words accepted with
probability arbitrarily close to 1?
This problem was proved undecidable recently;
to overcome this, several classes of probabilistic automata of different nature 
were proposed, for which the value 1 problem has been shown decidable.
In this paper, we introduce yet another class of probabilistic automata, called \textit{leaktight automata},
which strictly subsumes all classes of probabilistic automata whose value 1 problem is known to be decidable.

We prove that for leaktight automata, the value 1 problem is decidable (in fact, PSPACE-complete)
by constructing a saturation algorithm based on the computation of a monoid abstracting 
the behaviours of the automaton.
We rely on algebraic techniques developed by Simon to prove that this abstraction is complete.
Furthermore, we adapt this saturation algorithm to decide whether an automaton is leaktight.

Finally, we show a reduction allowing to extend
our decidability results from finite words to infinite ones,
implying that the value $1$ problem for probabilistic leaktight parity automata is decidable.
\end{abstract}

\maketitle\vfill

\section*{Introduction}

\medskip\textbf{Probabilistic automata.}
Rabin invented a very simple yet powerful
model of probabilistic machine called probabilistic automata,
which, quoting Rabin, ``are a generalization of finite deterministic automata''~\cite{R63}.
A probabilistic automaton has a finite set of states and reads input words
from a finite alphabet. 
The computation starts from the initial state and consists in reading the input word sequentially;
the state is updated according to transition probabilities 
determined by the current state and the input letter. 
The probability to accept a finite input word is the probability that the computation ends 
in one of the final states.

Probabilistic automata, and more generally partially observable Markov decision processes and stochastic games,
are a widely studied model of probabilistic machines used in many fields 
like software verification~\cite{BBG12,CDHR07}, image processing~\cite{CK97},
computational biology~\cite{DEKM99} and speech processing~\cite{M97}.
As a consequence, it is crucial to understand which decision problems are algorithmically tractable
for probabilistic automata.
From a language-theoretic perspective,
several algorithmic properties of probabilistic automata are known:
while language emptiness is undecidable~\cite{P71,B74,GO10},
functional equivalence is decidable~\cite{S61,T92}
as well as other properties~\cite{CMRR08}.

Our initial motivation for this work
comes from control and game theory:
we aim at solving algorithmic questions about partially
observable Markov decision processes and stochastic games.
For this reason, we consider
probabilistic automata as machines controlled by a blind controller,
who is in charge of choosing the sequence of input letters 
in order to maximize the acceptance probability.
While in a fully observable Markov decision
process the controller can observe the current state of the process to choose adequately 
the next input letter, 
a blind controller does not observe anything
and its choice depends only on the number of letters already chosen.
In other words, the strategy of a blind controller is an input word of the
automaton.

\medskip\textbf{The value of a probabilistic automaton.}
With this game-theoretic interpretation in mind,
we define the \emph{value} of a probabilistic
automaton as the supremum acceptance probability over all input words,
and we would like to compute this value.
Unfortunately, as a consequence of Paz undecidability result,
the value of an automaton is not computable in general.
However, the following decision
problem was conjectured by Bertoni~\cite{B74} to be decidable:

\smallskip
{\bf Value 1 problem:}
\emph{Given a probabilistic automaton, does it have value $1$?
In other words are there input words whose acceptance probability
is arbitrarily close to $1$?}

\smallskip
Recently, the second and fourth authors of the present paper proved that
the value~$1$ problem is undecidable~\cite{GO10}.

\medskip\textbf{Our result.}
We introduce a new class of probabilistic automata, called \emph{leaktight automata}, for which the value $1$ problem is decidable.
This subclass strictly subsumes all known subclasses of probabilistic automata sharing this decidability property
and has good closure properties.
Our algorithm to decide the value $1$ problem computes in polynomial space 
a finite monoid whose elements are directed graphs and checks whether
it contains a certain type of elements that are value $1$ witnesses.

\medskip\textbf{Related works.}
Introducing subclasses of probabilistic automata to cope with undecidability results has been
a fruitful and lively topic recently.
We discuss some of them here.

The first subclass which was introduced specifically to decide the value $1$ problem are the $\sharp$-acyclic
automata~\cite{GO10}.
Later on, Chatterjee and Tracol~\cite{CT12} introduced structurally simple automata, which are probabilistic automata 
satisfying a structural property (related to the decomposition-separation theorem from probability theory),
and proved that the value $1$ problem is decidable for structurally simple automata.
At the same time, a subset of the authors introduced leaktight automata, and proved a similar result.
As we shall see, both $\sharp$-acyclic and structurally simple automata are leaktight,
hence our results extend both~\cite{GO10} and~\cite{CT12}.

Quite recently, Chadha, Sistla and Viswanathan introduced the subclass of hierarchical automata~\cite{CSV11}, 
and showed that over infinite words, they recognize exactly the class of $\omega$-regular languages.
As we shall see, hierarchical automata are leaktight, hence as a consequence of our result, 
the value $1$ problem is decidable for hierarchical automata.

\medskip\textbf{Proof techniques.}
Our proof techniques totally depart from the ones used in~\cite{CSV11,CT12,GO10}.
We make use of algebraic techniques and in particular Simon's factorization forest theorem, 
which was used successfully to prove the decidability of the boundedness problem for distance automata~\cite{S94},
and extended models as desert automata and B-automata~\cite{K05,C09}

\medskip\textbf{Outline.}
Basic definitions are given in Section~\ref{sec:def}.

In Section~\ref{sec:algo}, we introduce the Markov monoid and the Markov monoid algorithm for the value $1$ problem;
since the problem is in general undecidable, the algorithm is incomplete:
a positive answer implies that the automaton has value $1$, but a negative answer gives no guarantee.

In Section~\ref{sec:leaktight}, we define the class of leaktight automata 
and show that the leaktight property is a sufficient condition
for this algorithm to be complete;
in particular, this implies that the value $1$ problem
is decidable for leaktight automata.

In Section~\ref{sec:leaktight_properties}, we show that 
the Markov monoid algorithm runs in polynomial space,
and obtain as a corollary that the value $1$ problem
for leaktight automata is $\PSPACE$-complete.
Furthermore, we extend the Markov monoid algorithm to check at the same time
whether an automaton is leaktight and whether in such case it has value $1$.

In Section~\ref{sec:leaktight_comparisons}, we further investigate the class of leaktight automata:
we provide examples of leaktight automata and show that all subclasses of probabilistic automata 
whose value $1$ problem is known to be decidable are leaktight.

In Section~\ref{sec:infinite}, we give a general theorem allowing to extend 
the decidability results from finite words to infinite words.

\section{Definitions}
\label{sec:def}
\subsection{Probabilistic automata}

We fix $A$ a finite alphabet. A (finite) word $u$ is a (possibly empty) sequence of letters $u = a_0 a_1 \cdots a_{n-1}$,
the set of finite words is denoted by $A^*$.
For $i \le j$ we denote by $u[i,j]$ the subword $a_i \cdots a_{j-1}$, and $u_{< p} = u[0,p] = a_0 a_1 \cdots a_{p-1}$.

Let $Q$ be a finite set of states. 
A probability distribution over $Q$ is a function $\delta : Q \rightarrow [0,1]$
such that $\sum_{q \in Q} \delta(q) = 1$;
we often see $\delta$ as a row vector of size $|Q|$.
We denote by $\frac{1}{3} \cdot q + \frac{2}{3} \cdot q'$ the distribution
that picks $q$ with probability $\frac{1}{3}$ and $q'$ with probability $\frac{2}{3}$,
and by $q$ the trivial distribution picking $q$ with probability $1$.
For a subset $R$ of states, the uniform distribution over $R$ picks each state in $R$
with probability $\frac{1}{|R|}$.
The support of a distribution $\delta$ is the set of states picked with positive probability,
\textit{i.e.} $\supp(\delta) = \set{q \in Q \mid \delta(q) > 0}$.
Finally, the set of probability distributions over $Q$ is $\DD(Q)$.

\begin{defi}[Probabilistic automaton]
A tuple $\AA = (Q, q_0, \Delta, F)$ represents a probabilistic automaton,
where $Q$ is a finite set of states, 
$q_0 \in Q$ is the initial state, $\Delta$ defines the transitions
and $F \subseteq Q$ is the set of accepting states.
\end{defi}

The transitions of a probabilistic automaton are given by a function 
$\Delta : Q \times A \rightarrow \DD(Q)$, where $\Delta(q,a)$ is the probability distribution obtained
by reading the letter $a$ from the state $q$.
The function $\Delta$ induces the function $\Delta' : \DD(Q) \times A \rightarrow \DD(Q)$,
where $\Delta'(\delta,a) = \sum_{q \in Q} \delta(q) \cdot \Delta(q,a)$.
Going further, $\Delta$ naturally extends to $\Delta^* : \DD(Q) \times A^* \rightarrow \DD(Q)$
by induction:
for a letter $a \in A$, we set $\Delta^*(\delta,a) = \Delta'(\delta,a)$,
and for an input word $u = a v$, we set $\Delta^*(\delta,u) = \Delta^*(\Delta'(\delta,a),v)$.
Intuitively, $\Delta^*(\delta,u)$ is the probability distribution obtained by reading the word $u$ starting
at the initial probability distribution $\delta$.
From now on, we will make no difference between $\Delta$, $\Delta'$ and $\Delta^*$, and denote the three of them by $\Delta$.

We denote by $\pRob{\AA}(s \xrightarrow{u} t)$ the probability to go from state $s$ to state $t$ reading $u$ on the automaton $\AA$,
\textit{i.e.} $\Delta(s,u)(t)$.
Then $\pRob{\AA}(s \xrightarrow{u} T)$ is defined as $\sum_{t \in T} \pRob{\AA}(s \xrightarrow{u} t)$.
Finally, the \emph{acceptance probability} of a word $u \in A^*$ by $\AA$ is $\pRob{\AA}(q_0 \xrightarrow{u} F)$,
which we denote by $\pRob{\AA}(u)$.

For computational purposes, we assume that each value is a rational number given by two integers
in binary decomposition.

\begin{defi}[Value]
The \emph{value} of a probabilistic automaton $\AA$, denoted by $\val{\AA}$,
is the supremum acceptance probability over all input words:
\begin{equation}
\label{eq:value}
\val{\AA} = \sup_{u \in A^*} \pRob{\AA}(u).
\end{equation}
\end{defi}

\subsection{The value~1 problem}
We are interested in the following decision problem:
\begin{problem}[Value~1 Problem]\label{prob:value1}
Given a probabilistic automaton $\AA$, decide whether $\val{\AA} = 1$.
\end{problem}

The value~1 problem can be reformulated using the notion
of \emph{isolated cut-point} introduced by Rabin in his seminal
paper~\cite{R63}: an automaton has value~1 if
and only if the cut-point~1 is \emph{not} isolated.


\begin{figure}[ht]
\begin{center}
\includegraphics[scale=1]{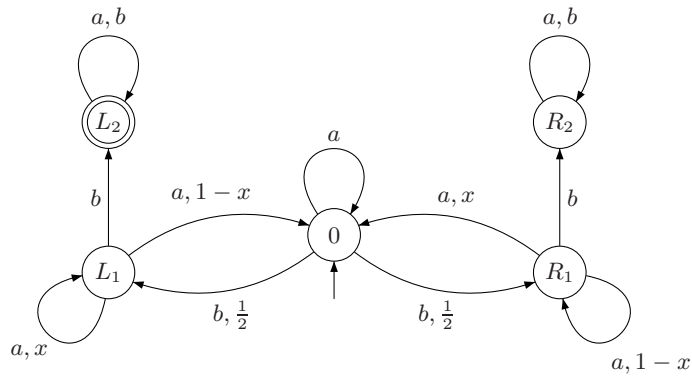}
\caption{\label{fig:x} This automaton has value~1 if and only if $x > \frac{1}{2}$.}
\end{center}
\end{figure}

The automaton depicted on figure~\ref{fig:x}
has value~1 if and only if $x > \frac{1}{2}$ (a similar example appears in~\cite{BBG12}).
The input alphabet is $A = \set{a,b}$,
the initial state is the central
state $0$ and the unique final state is $L_2$.

We describe the behaviour of this automaton.
After reading one $b$, the distribution is uniform over $L_1,R_1$.
To reach $L_2$, one needs to read a $b$ from the state $L_1$,
but on the right-hand side this leads to the non-accepting absorbing state $R_2$.
In order to maximize the probability to reach $L_2$,
one tries to ``tip the scales'' to the left.

If $x \leq \frac{1}{2}$, there is no hope to achieve this: 
reading a letter $a$ gives more chance to stay in $R_1$ 
than in $L_1$ thus all words are accepted with probability at most $\frac{1}{2}$,
and $\val{\AA} = \frac{1}{2}$.

However, if $x > \frac{1}{2}$ then we show that $\AA$ has value $1$.

We have:
\[
\pRob{\AA}(0 \xrightarrow{b a^n} L_1) = \frac{1}{2} \cdot x^n \qquad 
\textrm{ and } \qquad \pRob{\AA}(0 \xrightarrow{b a^n} R_1) = \frac{1}{2} \cdot (1-x)^n
\]

We fix an integer $N$ and analyse the action of reading $(b a^n)^N \cdot b$: 
there are $N$ ``rounds'',
each of them corresponding to reading $b a^n$ from $0$.
In a round, there are three outcomes: 
winning (that is, remaining in $L_1$) with probability $p_n = \frac{1}{2} \cdot x^n$,
losing (that is, remaining in $R_2$) with probability $q_n = \frac{1}{2} \cdot (1-x)^n$, 
or going to the next round (that is, reaching $0$) with probability $1 - (p_n + q_n)$.
If a round is won or lost, then the next $b$ leads to an accepting or rejecting sink; otherwise it goes on to the next round, for $N$ rounds. 
Hence:
$$\begin{array}{lll}
\pRob{\AA}((b a^n)^N \cdot b) & = & \sum_{k = 1}^N (1 - (p_n + q_n))^{k-1} \cdot p_n \\[1.5ex]
& = & p_n \cdot \frac{1 - (1 - (p_n + q_n))^N}{1 - (1 - (p_n + q_n))} \\[1.5ex]
& = & \frac{1}{1 + \frac{q_n}{p_n}} \cdot \left(1 - (1 - (p_n + q_n))^N\right) \\[1.5ex]
\end{array}$$

We now set $N = 2^n$.
A simple calculation shows that the sequence $((1 - (p_n + q_n))^{2^n})_{n \in \NN}$ converges to $0$ as $n$ goes to infinity.
Furthermore, if $x > \frac{1}{2}$ then $\frac{1-x}{x} < 1$, 
so $\frac{q_n}{p_n} = \left(\frac{1-x}{x}\right)^n$ converges to $0$ as $n$ goes to infinity.
It follows that the acceptance probability converges to $1$ as $n$ goes to infinity.
Consequently: $$\lim_n \pRob{\AA}((b a^n)^{2^n} \cdot b) = 1.$$
This example witnesses two surprising phenomena:
\begin{itemize}
	\item the value is discontinuous with respect to the transition probabilities,
as for $x = \frac{1}{2}$ the value is $\frac{1}{2}$, and for $x > \frac{1}{2}$ the value is $1$;
	\item the sequence of words $((b a^n)^{2^n} \cdot b)_{n \in \NN}$
witnessing the value $1$ involves two convergence speeds: indeed, the words $a^n b$
are repeated an exponential number of times, namely $2^n$.
One can show that repeating only $n$ times does not lead to words accepted with arbitrarily high probability.	
\end{itemize}

\subsection{Recurrent states and idempotent words}

We fix $\AA$ a probabilistic automaton, and define two main notions: recurrent states and idempotent words.

\begin{defi}[Induced Markov chain]
Let $u$ be a finite word, it induces a Markov chain $\MM_{\AA,u}$
whose state space is $Q$ and transition matrix $M_{\AA,u}$ is defined by:
\[
M_{\AA,u}(s,t) = \pRob{\AA}(s \xrightarrow{u} t).
\]
\end{defi}

We rely on the classical notion of recurrent states in Markov chains.

\begin{defi}[Recurrent state]
A state $s$ is \emph{$u$-recurrent} if it is recurrent in $\MM_{\AA,u}$.
\end{defi}


A finite word $u$ is \emph{idempotent} if reading once or twice the word
$u$ does not change qualitatively the transition probabilities.
 
\begin{defi}[Idempotent word]
A Markov chain is idempotent if its transition matrix $M$ satisfies
that for all states $s,t$:
\[
M(s,t) > 0 \iff M^2(s,t) > 0.
\]

A finite word $u$ is idempotent if $\MM_{\AA,u}$ is idempotent.
\end{defi}

In the case of idempotent words, recurrence of a state is easily characterized, relying on simple graph-theoretical arguments:

\begin{lem}
\label{lem:idempotent}
Let $u$ be an idempotent word.
A state $s$ is $u$-recurrent if and only if for all states $t$ we have:
\[
\MM_{\AA,u}(s,t) > 0 \implies \MM_{\AA,u}(t,s) > 0.
\]
\end{lem}


\section{An (incomplete) algorithm for the value 1 problem}
\label{sec:algo}
In this section, we present an algebraic algorithm for the value $1$ problem,
called the Markov monoid algorithm.
Since the problem is undecidable, this algorithm 
does not solve the problem on all instances;
we will show that it is \emph{correct},
\textit{i.e.} if it answers that an automaton has value $1$,
then the automaton does have value $1$,
but not \emph{complete},
\textit{i.e.} the converse does not hold.
In the next section, we shall show
that this algorithm is \emph{complete}
for the class of leaktight automata.

\subsection{The Markov monoid algorithm}
Our algorithm for the value $1$ problem computes iteratively
a set $\monoid$ of directed graphs called limit-words.
Each limit-word is meant to represent the asymptotic effect of
a sequence of input words, and some particular limit-words
can witness that the automaton has value $1$.

\begin{defi}[Limit-word]
A \emph{limit-word} is a function $\limu : Q^2 \to \set{0,1}$,
such that for all states $s$, there exists a state $t$
such that $\limu(s,t) = 1$.
\end{defi}
In proofs and examples, we will adopt either of the two equivalent views for limit-words:
graphs over the set $Q$ or square matrices over $Q \times Q$.

\begin{algorithm}[h!t]
\caption{The Markov monoid algorithm.}
\label{algo:markov_monoid}
\SetAlgoLined
\KwData{A probabilistic automaton.}
     
$\monoid \gets \set{\bolda \mid a\in A} \cup \set{\boldeps}$.

\Repeat{there is nothing to add}{
	\If{there is $\limu,\limv\in \monoid$ such that $\limu\cdot \limv \notin \monoid$}{
		add $\limu \cdot \limv$ to $\monoid$
		}

	\If{there is $\limu\in \monoid$ such that $\limu$ is idempotent	and $\limu^\sharp \notin \monoid$}{
		add $\limu^\sharp$ to $\monoid$
		}
}

\eIf{there is a value $1$ witness in  $\monoid$}{
	return true\;}{
	return false\;}{
}
\end{algorithm}

We now explain the algorithm in detail.
For the remainder of this section, we fix $\AA$ a probabilistic automaton.
Initially, $\monoid$ only contains those limit-words $\bolda$
that are induced by input letters $a \in A$ :
\[
\forall s,t \in Q,\ (\bolda(s,t) = 1 \iff \pRob{\AA}(s \xrightarrow{a} t) > 0),
\]
plus the limit-word $\boldeps$ which is induced by the empty word:
\[
\forall s,t \in Q,\ (\boldeps(s,t) = 1 \iff s = t).
\]

The algorithm repeatedly adds new limit-words to $\monoid$.
There are two ways for that: concatenating two limit-words or iterating an idempotent limit-word.

\vskip1em
\textbf{Concatenation of two limit-words}
The \emph{concatenation} of two limit-words $\limu$
and $\limv$ is the limit-word $\limu \cdot \limv$ such that:
\[
(\limu \cdot \limv)(s,t) = 1 \iff \exists q \in Q,\
\limu(s,q) = 1 \text{ and } \limv(q,t) = 1.
\]
In other words, concatenation corresponds to the multiplication
of matrices with coefficients in the boolean semiring
$(\{0,1\},\vee,\wedge)$.
Intuitively, the concatenation of two limit-words corresponds to
the concatenation of two sequences $(u_n)_{n\in\NN}$ and $(v_n)_{n\in\NN}$ of input
words into the sequence $(u_n\cdot v_n)_{n\in\NN}$.

We say that a limit-word $\limu$ is idempotent if $\limu \cdot \limu = \limu$.
The following lemma gives simple properties of idempotent limit-words.

\begin{lem}\label{lem:basic_limit_words}
For all limit-words $\limu$:
\begin{itemize}
	\item the limit-word $\limu^{|Q|!}$ is idempotent,
	\item if $\limu$ is idempotent, then 
for all states $r \in Q$, there exists a state $r' \in Q$ such that $\limu(r,r') = 1$ and $r'$ is $\limu$-recurrent.
\end{itemize}
\end{lem}

\noindent The proof is omitted and relies on simple graph-theoretical arguments.

\vskip1em
\textbf{Iteration of an idempotent limit-word}
The \emph{iteration} $\limu^\sharp$ of a limit-word $\limu$ is only defined
when $\limu$ is idempotent.
It relies on the notion of $\limu$-recurrent state.
\begin{defi}[$\limu$-recurrence]
Let $\limu$ be an idempotent limit-word.
A state $s$ is $\limu$-recurrent if for all states $t$,
we have:
\[
\limu(s,t)=1 \implies \limu(t,s) = 1.
\]
\end{defi}

\noindent Note that this echoes Lemma~\ref{lem:idempotent}.
The \emph{iterated limit-word} $\limu^\sharp$ removes from $\limu$
any edge that does not lead to a recurrent state:
\[
\limu^\sharp(s,t) = 1 \iff \limu(s,t) = 1 \text{ and } t \text{ is } \limu\text{-recurrent}.
\]
Intuitively, if a limit-word $\limu$ represents a
sequence $(u_n)_{n\in\NN}$ then its iteration $\limu^\sharp$
represents the sequence $\left(u_n^{f(n)}\right)_{n\in\NN}$
for some increasing function $f:\NN\to\NN$.

\subsection{The Markov monoid and value \texorpdfstring{$1$}{1} witnesses}

The set of limit-words $\monoid$ computed by Algorithm~\ref{algo:markov_monoid}
is called the Markov monoid.

\begin{defi}[Markov monoid]
The Markov monoid associated with $\AA$ is the smallest set of limit-words containing
$\set{\bolda \mid a \in A} \cup \set{\boldeps}$ and closed under concatenation and iteration.
\end{defi}

Two key properties, \emph{consistency} and \emph{completeness},
ensure that the limit-words of the Markov monoid
reflect exactly every possible asymptotic effect
of a sequence of input words.

\begin{defi}[Reification]
A sequence $(u_n)\nNN$ of words reifies a limit-word $\limu$ if
for all states $s,t$, $(\pRob{\AA}(s \xrightarrow{u_n} t))\nNN$ converges and:
\begin{equation}
\label{eq:consistency}
\limu(s,t) = 1 \iff \lim_n \pRob{\AA}(s \xrightarrow{u_n} t) > 0.
\end{equation}
\end{defi}

Note that if $(u_n)\nNN$ reifies $\limu$, then any subsequence of $(u_n)\nNN$ also does.
We will use this simple observation several times.

\begin{defi}[Consistency]
\label{def:consistency}
A set of limit-words $\monoid$ is \emph{consistent}
with $\AA$ if for every limit-word $\limu \in \monoid$, 
there exists a sequence of input words $(u_n)\nNN$ which reifies $\limu$.
\end{defi}

\begin{defi}[Completeness]
\label{def:complete}
A set of limit-words $\monoid$ is \emph{complete}
for $\AA$ if
for each sequence of input words
$(u_n)_{\nNN}$, there exists $\limu \in \monoid$
such that for all states $s,t \in Q$:
\begin{equation}
\label{eq:completeness}
\limsup_n \pRob{\AA}(s \xrightarrow{u_n} t) = 0 \implies \limu(s,t) = 0.
\end{equation} 
\end{defi}

Limit-words are useful to decide the value $1$ problem
because some of these are witnesses that the automaton has value $1$.

\begin{defi}[Value $1$ witness]
\label{def:value1witnesses}
A \emph{value $1$ witness} is a limit-word $\limu$ such that
for all states $t$:
\begin{equation}
\label{eq:witness}
\limu(q_0,t) = 1 \implies t \in F,
\end{equation}
where $q_0$ is the initial state of the automaton.
\end{defi}

Thanks to value $1$ witnesses,
the answer to the value $1$ problem can be read in a
consistent and complete set of limit-words:

\begin{lem}[A criterion for value $1$]
\label{lem:criterion}
If $\monoid$ is consistent with $\AA$ and complete for $\AA$,
then $\AA$ has value $1$ if and only if $\monoid$ contains a value $1$ witness.

Specifically:
\begin{itemize}
	\item If $\monoid$ is consistent with $\AA$ and contains a value $1$ witness,
	then $\AA$ has value $1$,
	\item If $\monoid$ is complete for $\AA$ and $\AA$ has value $1$,
	then $\AA$ contains a value $1$ witness.
\end{itemize}
\end{lem}

\begin{proof}
We prove the first item.
Assume that $\monoid$ is consistent with $\AA$ and contains a value $1$ witness $\limu$.
Since $\monoid$ is consistent,
there exists a sequence $(u_n)\nNN$ reifying $\limu$.
It follows from~\eqref{eq:consistency} and~\eqref{eq:witness}
that for $t \notin F$, we have $\lim_n \pRob{\AA}(q_0 \xrightarrow{u_n} t) = 0$.
Thus $\lim_n \pRob{\AA}(u_n) = \sum_{t \in F} \lim_n \pRob{\AA}(q_0 \xrightarrow{u_n} t) = 1$,
so $\AA$ has value $1$.

We now prove the second item.
Assume that $\monoid$ is complete for $\AA$ and that $\AA$ has value $1$.
Then there exists a sequence of words $(u_n)\nNN$
such that $\lim_n \pRob{\AA}(u_n) = 1$,
\textit{i.e.} $\lim_n \sum_{t \in F} \pRob{\AA}(q_0 \xrightarrow{u_n} t) = 1$.
Since for all $n \in \NN$, we have $\sum_{q \in Q} \pRob{\AA}(q_0 \xrightarrow{u_n} q) = 1$,
then for all $t \notin F$, $\limsup_n \pRob{\AA}(q_0 \xrightarrow{u_n} t) = 0$.
Since $\monoid$ is complete,
there exists a limit-word $\limu$ such that~\eqref{eq:completeness} holds.
Then $\limu$ is a value $1$ witness:
let $t \in Q$ such that $\limu(q_0,t) = 1$, 
then according to~\eqref{eq:completeness},
$\limsup_n \pRob{\AA}(q_0 \xrightarrow{u_n} t) > 0$, hence $t \in F$.
\end{proof}

\subsection{Correctness of the Markov monoid algorithm}
\label{subsec:correctness}

\begin{thm}
\label{theo:consistency}
The Markov monoid associated with $\AA$ is consistent.
\end{thm}

This implies that if the Markov monoid algorithm outputs ``true'',
then for sure the input automaton has value $1$. 
This positive result holds for every automaton (leaktight or not).

To prove Theorem~\ref{theo:consistency}, 
recall that the Markov monoid is the smallest set of limit-words
containing $\set{\bolda \mid a \in A} \cup \set{\boldeps}$ and closed
under concatenation and iteration,
hence it suffices to prove that the initial elements form
a consistent set, and the closure under the two operations.

First, $\bolda$ is reified by the constant sequence $(a)\nNN$,
and $\boldeps$ by the constant sequence $(\varepsilon)\nNN$.
We state the closure under the two operations in the following proposition:

\begin{prop}\label{prop:consistency}
Let $(u_n)\nNN$ and $(v_n)\nNN$ be two sequences that reify the limit-words $\limu$ and $\limv$ respectively. 
Then:
\begin{enumerate}
	\item the sequence of words $(u_n \cdot v_n)\nNN$ reifies $\limu \cdot \limv$,
	\item if $\limu$ is idempotent, then there exists an increasing function $f : \NN \to \NN$ such that 
	for all increasing functions $g : \NN \to \NN$ satisfying $g \ge f$,
	the sequence $\left(u_{g(n)}^n\right)\nNN$ reifies the limit-word $\limu^\sharp$.
\end{enumerate}
\end{prop}

\noindent The statement about iteration is stronger than required: 
the existence of $f$ such that $(u_{f(n)}^n)\nNN$ reifying the limit-word $\limu^\sharp$ is enough to prove Theorem~\ref{theo:consistency}.
However, we will use this stronger result later on (in Section~\ref{subsec:simple_automata}).

\begin{proof}\hfill
\begin{enumerate}
	\item Let $w_n = u_{n} \cdot v_{n}$.
Then $(w_n)\nNN$ reifies $\limu\cdot \limv$, since:
\[
\pRob{\AA}(s \xrightarrow{w_n} t) = 
\sum_{r \in Q} \pRob{\AA}(s \xrightarrow{u_n} r) \cdot \pRob{\AA}(r \xrightarrow{v_n} t).
\]

	\item Consider the Markov chain $\MM$ with state space $Q$ and transition matrix $M$ defined by $M(s,t) = \lim_n \pRob{\AA}(s \xrightarrow{u_n} t)$.
Since $(u_n)\nNN$ reifies $\limu$, we have $\limu(s,t) = 1$ if and only if $M(s,t) > 0$.
First observe that since $\limu$ is idempotent, the Markov chain $\MM$ is aperiodic.
According to standard results about finite Markov chains,
this implies that the sequence of matrices $(M^k)_{k \in \NN}$ 
has a limit which we denote by $M^\infty$, satisfying the following:
\begin{equation}
\label{eq:trans}
\forall s,t \in Q,\ M^\infty(s,t) > 0 \implies t \text{ is recurrent in } \MM.
\end{equation} 
By definition the sequence of matrices $\left(M_{\AA,u_n}\right)\nNN$ converges to $M$.
Since the matrix product operation is continuous, for every $k \in \NN$,
$\left(M_{\AA,u_n}^k\right)\nNN$ converges to $M^k$.
So for every $k \ge 1$, there exists $N_k \in \NN$ such that for all $p \ge N_k$,
$||M^k - M_{\AA,u_p}^k ||_\infty \leq \frac{1}{k}$.
We define $f : \NN \to \NN$ by induction, so that $f(k)$
is the maximum of $f(k-1) + 1$ and of $N_k$,
ensuring that $f$ is increasing.
Then for any increasing function $g : \NN \to \NN$ satisfying $g \ge f$,
the sequence of matrices $\left(M_{\AA,u_{g(n)}}^n\right)\nNN$ converges to $M^\infty$.
We prove that $\left(u_{g(n)}^n\right)\nNN$ reifies $\limu^\sharp$:
\begin{align*}
\limu^\sharp(s,t) = 1
& \iff \limu(s,t) = 1 \text{ and } t \text{ is } \limu\text{-recurrent}\\
& \iff M(s,t) > 0 \text{ and } t \text{ is } \text{recurrent in } \MM\\
& \iff M^\infty(s,t) > 0 \\
& \iff \lim_n \pRob{\AA}(s \xrightarrow{u_{g(n)}^n} t) > 0,
\end{align*}
where the first equivalence is by definition of the iteration,
the second holds because $(u_n)\nNN$ reifies $\limu$,
the third by definition of $M^\infty$,
and the fourth because $\left(M_{\AA,u_{g(n)}^n}\right)\nNN$ converges to $M^\infty$.
\end{enumerate}
This concludes the proof.
\end{proof}

Note that completeness is not true in general; for instance,
one can show that the Markov monoid of the automaton represented 
in figure~\ref{fig:x} is not complete.
The next section gives a sufficient condition for completeness: the leaktight property.

\section{Decidability of the value 1 problem for leaktight automata}
\label{sec:leaktight}
In this section we establish our main result:

\begin{thm}
The value $1$ problem is decidable for leaktight automata.
\end{thm}

The definition of leaktight automata is given in the next subsection.
For now (in this section), we are only interested in decidability issues;
we will actually prove in Section~\ref{sec:leaktight_properties} that the value $1$ problem
is $\PSPACE$-complete for leaktight automata.

Note that as observed in the literature~\cite{BBG12,CSV13,Fijalkow14}, the value $1$ problem for probabilistic automata over finite words 
is equivalent to the emptiness problem for probabilistic B\"uchi automata with positive semantics,
hence we obtain the following corollary:

\begin{cor}
The emptiness problem is decidable for probabilistic B\"uchi leaktight automata with positive semantics.
\end{cor}

The following theorem proves that the Markov monoid of a 
leaktight automaton is complete;
since it is always consistent,
by Lemma~\ref{lem:criterion}, the Markov monoid algorithm solves the value $1$ problem
for leaktight automata.

\begin{thm}\label{theo:completeness}
If a probabilistic automaton is leaktight then its Markov monoid is complete. 
\end{thm}

The remainder of this section is devoted to the proof
of Theorem~\ref{theo:completeness}.
We first define the leaktight property,
and extend the Markov monoid.
This extended version allows to 
state an algebraic characterization of the leaktight property.
Then, the technical core of the proof relies on a subtle algebraic argument
based on the existence of $\sharp$-factorization trees of bounded height~\cite{S90,S94,C09,T11}.

\subsection{Leaks}
\label{subsec:leaks}

The undecidability of the value~1 problem
comes from the necessity to compare
parallel convergence rates
in order to track down vanishing probabilities.
Comparing two convergence rates may require to compare
the decimals of the rates up to an arbitrary
precision,
which in turn can encode a Post correspondence problem,
hence the undecidability.

One of the phenomena that makes tracking vanishing probabilities
difficult are \emph{leaks}.
A leak occurs in an automaton
when a sequence of words turns a set of states $C\subseteq Q$
into a recurrence class $C$ \emph{on the long run},
but \emph{on the short run}, some of the probability of the recurrence class
is ``leaking'' to a \emph{different} recurrence class.

\begin{figure}[ht]
\begin{center}
\includegraphics[scale=1]{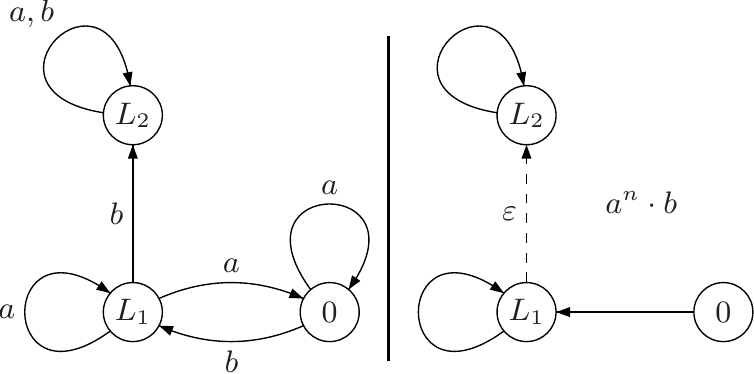}
\caption{\label{fig:ex_leak} $(a^n \cdot b)\nNN$ is a leak from $L_1$ to $L_2$.}
\end{center}
\end{figure}

Such leaks occur in the automaton depicted in the left hand side of figure~\ref{fig:ex_leak}
with the input sequence $(a^nb)\nNN$.
As $n$ grows large, the probability
to reach $L_2$ from $L_1$ while reading the input word $a^nb$
vanishes, thus the sets $\set{L_1}$ and $\set{L_2}$
are two different recurrence classes on the long run (\textit{i.e.} asymptotically),
however on the short run remains a small yet positive probability
to reach $L_2$ from $L_1$.

The right hand side of figure~\ref{fig:ex_leak} shows the asymptotic behaviour
of reading $(a^nb)\nNN$.

Since the automaton in figure~\ref{fig:x} contains two symmetric parts 
identical to figure~\ref{fig:ex_leak}, it features one leak on the left hand side
and another in the right hand side.
As a consequence, the real asymptotic behaviour is
complex and depends on the compared speeds of these leaks.

An automaton without leak is called a leaktight automaton.
In this section we prove that the value~1 problem
is decidable when restricted to the subclass
of leaktight automata.

The formal definition of a leak is as follows:

\begin{defi}[Leak]
\label{def:leak}
Let $(u_n)\nNN$ be a sequence of idempotent words.
Assume that the sequence of matrices $\pRob{\AA}(u_n)$ converges to a limit $M$,
that this limit is idempotent and denote $\MM$ the assocaited Markov chain.

The sequence $(u_n)\nNN$ is a leak if there exist $r,q \in Q$ such that the following three conditions hold:
\begin{enumerate}
	\item $r$ and $q$ are recurrent in $\MM$,
	\item $\lim_n \pRob{\AA}(r \xrightarrow{u_n} q) = 0$,
	\item for all $n \in \NN$, $\pRob{\AA}(r \xrightarrow{u_n} q) > 0$.
\end{enumerate}
\end{defi}


\begin{defi}[Leaktight automata]
A probabilistic automaton is leaktight if it has no leak.
\end{defi}

Several examples of leaktight automata are given in Section~\ref{sec:leaktight_comparisons}.

\subsection{The extended Markov monoid}

The existence of leaks can be decided by
a slight extension of the Markov monoid algorithm which keeps track
of strictly positive transition probabilities.

\begin{defi}[Extended limit-word]
An \emph{extended limit-word} is a couple $(\limu,\limu_+)$
of two limit-words, such that for all $s,t \in Q$, we have $\limu(s,t) = 1 \Longrightarrow \limu_+(s,t) = 1$.
\end{defi}
As for limit-words, extended limit-words can be seen either as 
graphs over the set $Q$, or couples of square matrices over $Q \times Q$.
Such a graph has two different kind of edges: an edge $(s,t)$ is ``normal'' if $\limu(s,t) = 1$,
and is a $+$-edge if $\limu(s,t) = 0$ but $\limu_+(s,t) = 1$.

We define the concatenation and iteration operations for extended limit-words.
The \emph{concatenation} of two extended limit-words 
$(\limu,\limu_+)$ and $(\limv,\limv_+)$ is the component-wise concatenation,
\textit{i.e.} $(\limu \cdot \limv, \limu_+ \cdot \limv_+)$.
The \emph{iteration} of an extended limit-word $(\limu,\limu_+)$ is only defined
when it is idempotent (\textit{i.e.} component-wise idempotent),
by $(\limu,\limu_+)^\sharp = (\limu^\sharp,\limu_+)$.

\begin{defi}[Extended Markov monoid]
The extended Markov monoid is the smallest set of extended limit-words containing
$\set{(\bolda,\bolda) \mid a \in A} \cup \set{(\boldeps,\boldeps)}$ 
and closed under concatenation and iteration.
\end{defi}

Note that if $(\limu,\limu_+)$ is in the extended Markov monoid,
then $\limu$ is in the Markov monoid.

The essential difference between the Markov monoid and its extended version
is that the extension keeps track of those edges that are deleted by
successive iteration operations. 
This serves two purposes: first, to characterize
the leaktight property in algebraic terms, and second,
to prove Theorem~\ref{theo:completeness}.

We state a consistency result for the extended Markov monoid,
extending Theorem~\ref{theo:consistency}.
The proofs of both these results are similar and
given only once.

\begin{lem}\label{lem:consistency_extended}
For each $(\limu,\limu_+)$ in the extended Markov monoid,
there exists a sequence $(u_n)_{n\in\NN}$ such that
for all states $s,t \in Q$, $(\pRob{\AA}(s \xrightarrow{u_n} t))\nNN$ converges and:
\begin{align}
\label{eq:appcond1}
& \limu(s,t) = 1 \iff \lim_n \pRob{\AA}(s \xrightarrow{u_n} t) > 0,\\
\label{eq:appcond2}
& \text{for all } n \in \NN,\ \left(\limu_+(s,t) = 1 \iff \pRob{\AA}(s \xrightarrow{u_n} t) > 0\right).
\end{align} 
\end{lem}

\subsection{Leak witnesses}

\begin{defi}[Leak witness]
\label{eq:leakwitness}
An idempotent extended limit-word $(\limu,\limu_+)$ is a \emph{leak witness} if there exist $r,q \in Q$ such that
the following three conditions hold:
\begin{enumerate}
  \item $r$ and $q$ are $\limu$-recurrent,
  \item $\limu(r,q) = 0$,
  \item $\limu_+(r,q) = 1$.
\end{enumerate}
\end{defi}

\begin{lem}\label{lem:leaktight_characterization_direct}
If a probabilistic automaton is leaktight, then its extended Markov monoid does not contain
any leak witness.
\end{lem}

\begin{proof}
Suppose that there is a leak witness $(\limu,\limu_+)$ in the extended Markov monoid:
$\limu$ and $\limu_+$ are idempotent and there exists $r,q \in Q$
such that $r$ and $q$ are $\limu$-recurrent, $\limu(r,q) = 0$ and $\limu_+(r,q) = 1$.
We prove that there exists a leak.

Thanks to Lemma~\ref{lem:consistency_extended},
there exists a sequence $(u_n)\nNN$ 
satisfying~\eqref{eq:appcond1} and~\eqref{eq:appcond2}.
Note that since $\limu_+$ is idempotent,~\eqref{eq:appcond2}
implies that for all $n \in \NN$, $u_n$ is idempotent.

Consider the Markov chain $\MM$ with state space $Q$ and transition matrix $M$ defined by $M(s,t) = \lim_n \pRob{\AA}(s \xrightarrow{u_n} t)$.
$\MM$ is idempotent since $\limu$ is idempotent and thanks to~\eqref{eq:appcond1}.

We show that $(u_n)\nNN$ is a leak. There are three conditions to be met.

First, $r$ and $q$ are recurrent in $\MM$: this follows from~\eqref{eq:appcond1}
and the fact that $r$ and $q$ are $\limu$-recurrent.
Second, $\lim_n \pRob{\AA}(r \xrightarrow{u_n} q) = 0$: this follows from~\eqref{eq:appcond1}
and the fact that $\limu(r,q) = 0$.
Third, for all $n \in \NN$, $\pRob{\AA}(r \xrightarrow{u_n} q) > 0$: this follows from~\eqref{eq:appcond2}
and the fact that $\limu_+(r,q) = 1$.
\end{proof}

As we will show in the next section, the converse of Lemma~\ref{lem:leaktight_characterization_direct}
is also true, which gives an algebraic characterization of the leaktight property
using the extended Markov monoid.
However, the proof of the converse implication is more involved and requires
the lower bound lemma (Lemma~\ref{lem:lowerbound}),
which is the object of the next subsection.

\subsection{Stabilization monoids and \texorpdfstring{$\sharp$}{sharp}-factorization trees}

We now introduce the technical material required to state and prove the lower bound lemma.
The key notions here are stabilization monoids and $\sharp$-factorization trees.

Factorization trees for monoids have been introduced by Simon~\cite{S90}.
Roughly speaking, Simon's factorization theorem
states that given a morphism $\phi : A^* \rightarrow M$
from the set of finite words over $A$ to a finite monoid $M$,
the following holds:
for all words $u$, the computation of $\phi(u)$
can be factorized in a tree whose depth
is bounded independently of the length of the word.

Simon later developed the notion of decomposition trees to solve
the limitedness problem for distance automata~\cite{S94}.
To this end, he defined an iteration operation $\sharp$
for monoids over the tropical semiring $(\NN \cup \set{\infty},\min,+)$.
Then Kirsten extended this technique to desert automata and the nested distance desert automata~\cite{K05}.
After him, Colcombet generalized this approach by defining stabilization monoids~\cite{C09}, 
which are monoids equipped with an iteration operation, and proved 
the existence of $\sharp$-factorization trees of bounded depth.
The formal definition is as follows:

\begin{defi}[Stabilization monoid]
A stabilization monoid $(M,\cdot,\sharp)$ is a finite monoid $(M,\cdot)$
equipped with an iteration operation $\sharp : E(M) \rightarrow E(M)$,
where $E(M)$ is the set of idempotents of $M$, such that:
\begin{align}
\label{eq:stab_monoid1}
& (a \cdot b)^\sharp \cdot a = a \cdot (b \cdot a)^\sharp \qquad 
\text{ for } a \cdot b \in E(M) \text{ and } b \cdot a \in E(M),\\
\label{eq:stab_monoid2}
& (e^\sharp)^\sharp = e^\sharp \qquad \qquad \qquad \qquad \text{ for } e \in E(M),\\
\label{eq:stab_monoid3}
& e^\sharp \cdot e = e^\sharp \qquad \qquad \qquad \qquad \text{ for } e \in E(M).
\end{align}  
\end{defi}

\begin{lem}\label{lem:stabilization_monoid}
The extended Markov monoid is a stabilization monoid.
\end{lem}

\begin{proof}
To start with, the extended Markov monoid is a monoid for the concatenation: $\boldeps$
is the neutral element, and the concatenation is associative.

Now, let us prove the three properties required for the iteration operation $\sharp$.

Proof of~\eqref{eq:stab_monoid1}. 
Let $(\limu,\limu_+),(\limv,\limv_+)$ such that 
$(\limu \cdot \limv,\limu_+ \cdot \limv_+)$ and $(\limv \cdot \limu,\limv_+ \cdot \limu_+)$ are idempotent.
By definition $\left((\limu,\limu_+) \cdot (\limv,\limv_+)\right)^\sharp \cdot (\limu,\limu_+)$
is equal to $\left((\limu \cdot \limv)^\sharp \cdot \limu, \limu_+ \cdot \limv_+ \cdot \limu_+\right)$,
and $(\limu,\limu_+) \cdot \left((\limv,\limv_+) \cdot (\limu,\limu_+)\right)^\sharp$
to $\left(\limu \cdot (\limv \cdot \limu)^\sharp, \limu_+ \cdot \limv_+ \cdot \limu_+\right)$.
Let $s,t \in Q$, we have the following equivalence: $\left((\limu \cdot \limv)^\sharp \cdot \limu\right)(s,t) = 1$
if and only if:
\begin{align}
\label{eq:stabilization_monoid1}
\text{there exists } r,q \in Q,\ \limu(s,r) = 1 \wedge \limv(r,q) = 1 \wedge \limu(q,t) = 1 \wedge 
q \text{ is } (\limu \cdot \limv) \text{-recurrent},
\end{align}
and similarly, $\left(\limu \cdot (\limv \cdot \limu)^\sharp\right)(s,t) = 1$
if and only if:
\begin{align}
\label{eq:stabilization_monoid2}
\text{there exists } r,q \in Q,\ \limu(s,r) = 1 \wedge \limv(r,q) = 1 \wedge \limu(q,t) = 1 \wedge 
t \text{ is } (\limv \cdot \limu) \text{-recurrent}.
\end{align}
We show that~\eqref{eq:stabilization_monoid1} and~\eqref{eq:stabilization_monoid2} are equivalent.
Assume~\eqref{eq:stabilization_monoid1}, and prove that $t$ is $(\limv \cdot \limu)$-recurrent.
Let $p \in Q$ such that $(\limv \cdot \limu)(t,p) = 1$.
Since $\limv$ is a limit-word, there exists $\ell \in Q$ such that
$\limv(p,\ell) = 1$.
Observe that $\limu(q,t) = 1$, $(\limv \cdot \limu)(t,p) = 1$ and $\limv(p,\ell) = 1$,
so $(\limu \cdot \limv)^2(q,\ell) = 1$.
As $\limu \cdot \limv$ is idempotent, this implies $(\limu \cdot \limv)(q,\ell) = 1$.
Since $q$ is $(\limu \cdot \limv)$-recurrent, we have $(\limu \cdot \limv)(\ell,q) = 1$.
Altogether, $\limv(p,\ell) = 1$, $(\limu \cdot \limv)(\ell,q) = 1$ and $\limu(q,t) = 1$
imply that $(\limv \cdot \limu)^2(p,t) = 1$.
As $\limv \cdot \limu$ is idempotent, this implies $(\limv \cdot \limu)(p,t) = 1$,
so $t$ is $(\limv \cdot \limu)$-recurrent, and~\eqref{eq:stabilization_monoid2} is proved.
Conversely, assume~\eqref{eq:stabilization_monoid2}, and prove that $q$ is $(\limu \cdot \limv)$-recurrent.
Note that from $\limv(r,q) = 1$, $\limu(q,t) = 1$ and the fact that $t$ is $(\limv \cdot \limu)$-recurrent,
we obtain that $(\limv \cdot \limu)(t,r) = 1$.
Let $p \in Q$ such that $(\limu \cdot \limv)(q,p) = 1$.
Since $\limu$ is a limit-word, there exists $\ell \in Q$ such that
$\limu(p,\ell) = 1$.
Observe that $\limv(r,q) = 1$, $(\limu \cdot \limv)(q,p) = 1$ and $\limu(p,\ell) = 1$,
so $(\limv \cdot \limu)^2(r,\ell) = 1$, and with $(\limv \cdot \limu)(t,r) = 1$
this implies $(\limv \cdot \limu)^3(t,\ell) = 1$.
As $\limv \cdot \limu$ is idempotent, this implies $(\limv \cdot \limu)(t,\ell) = 1$.
Since $t$ is $(\limv \cdot \limu)$-recurrent, we have $(\limv \cdot \limu)(\ell,t) = 1$.
Altogether, $\limu(p,\ell) = 1$, $(\limv \cdot \limu)(\ell,t) = 1$, 
$(\limv \cdot \limu)(t,r) = 1$ and $\limv(r,q) = 1$
imply that $(\limu \cdot \limv)^3(p,q) = 1$.
As $\limu \cdot \limv$ is idempotent, this implies $(\limu \cdot \limv)(p,q) = 1$,
so $q$ is $(\limu \cdot \limv)$-recurrent, and~\eqref{eq:stabilization_monoid1} is proved.
The property~\eqref{eq:stab_monoid1} follows.

Proof of~\eqref{eq:stab_monoid2}. 
This boils down to proving $(\limu^\sharp)^\sharp = \limu^\sharp$.
This is clear from the definition of $\limu^\sharp$, since the notions
of $\limu$-recurrence and $\limu^\sharp$-recurrence coincide.

Proof of~\eqref{eq:stab_monoid3}. 
This boils down to proving $\limu^\sharp \cdot \limu = \limu^\sharp$.
It follows from the observation that if $r \in Q$ is $\limu$-recurrent
and $\limu(r,t) = 1$, then $t$ is $\limu$-recurrent
(under the assumption that $\limu$ is idempotent).
\end{proof}

\begin{defi}
\label{def:sharp_factorization_tree}
Let $A$ be a finite alphabet, $(M,\cdot,\sharp)$ a stabilization monoid
and $\phi: A^* \to M$ a morphism into the submonoid $(M,\cdot)$. 
A $\sharp$-\emph{factorization tree} of a word $u \in A^*$ 
is a finite unranked ordered tree, whose nodes have labels in $A^* \times M$ and such that:
\begin{enumerate}[label=\roman*)]
	\item the root is labelled by $(u,\limu)$, for some $\limu \in M$,
	\item every internal node with two children (called \emph{concatenation} nodes) labelled by 
$(u_1,\limu_1)$ and $(u_2,\limu_2)$ is labelled by 
$(u_1 \cdot u_2,\limu_1 \cdot \limu_2)$,
	\item every internal node with three or more children (called \emph{iteration} nodes) 
is labelled by $(u_1 \ldots u_n,\bolde^\sharp)$ for some $\bolde \in E(M)$,
and its children are labelled by $(u_1,\bolde),\ldots,(u_n,\bolde)$.
	\item[iv)] every leaf is labelled by $(a,\bolda)$ where $a$ is a letter,
or $(\varepsilon,\boldeps)$.
\end{enumerate}
\end{defi}

\noindent Note that in a factorization tree, the second label is not always the image of the first component under $\phi$;
indeed, it is an element of the stabilization monoid $(M,\cdot,\sharp)$ whereas the image of a finite word under $\phi$
is an element of the submonoid $(M,\cdot)$.
However, the projection of second label into this submonoid (which consists in ignoring the operation $\sharp$)
is indeed the image of the first component under $\phi$.

The following theorem was stated for the tropical semiring in~\cite{S94}, and generalized in~\cite{C09}. A simple proof can be found in~\cite{T11}.
 
\begin{thm}\label{theo:decomposition}
Let $A$ be a finite alphabet, $(M,\cdot,\sharp)$ a stabilization monoid
and $\phi: A^* \to M$ a morphism into the submonoid $(M,\cdot)$.
Every word $u \in A^*$ has a $\sharp$-factorization tree whose depth is less than $3 \cdot |M|$.
\end{thm}

\subsection{The lower bound lemma}

We are ready to
state and prove the lower bound lemma, which is the central argument in
the proof of completeness of leaktight Markov monoids.

\begin{lem}[Lower bound lemma]\label{lem:lowerbound}
Let $\AA$ be a probabilistic automaton whose
extended Markov monoid contains no leak witness.
Let $\pmin$ the smallest non-zero transition probability of $\AA$.
Then for all words $u \in A^*$, there exists
$(\limu,\limu_+)$ in the extended Markov monoid
such that, for all states $s,t$:
\begin{align}
\label{eq:lb1}
\limu_+(s,t) = 1 & \iff \pRob{\AA}(s \xrightarrow{u} t) > 0,\\
\label{eq:lb2}
\limu(s,t) = 1 & \implies  \pRob{\AA}(s \xrightarrow{u} t) \geq \dabound.
\end{align}
\end{lem}

\begin{proof}
Consider a finite word $u \in A^*$;
by Theorem~\ref{theo:decomposition} applied to the extended Markov monoid $\monoid_+$ associated with $\AA$
(which is a stabilization monoid thanks to Lemma~\ref{lem:stabilization_monoid})
and the morphism $\phi: A \to M$ defined by $\phi(a) = (\bolda,\bolda)$,
there exists a $\sharp$-factorization tree of depth at most $3 \cdot |\monoid_+|$,
whose root is labelled by $(u,(\limu,\limu_+))$ for some extended limit-word $(\limu,\limu_+)$.

The depth of a node in this tree is defined in a bottom-up fashion:
the leaves have depth zero, and a node has depth one plus the maximum of the depths of its children.

We prove by a bottom-up induction (on $h$) that for every node $(u,(\limu,\limu_+))$ 
of this tree at depth $h$, for all states $s,t$:
\begin{align}
\label{eq:induc1}
\limu_+(s,t)=1 & \iff \pRob{\AA}(s \xrightarrow{u} t) > 0,\\
\label{eq:induc2}
\limu(s,t) = 1 & \implies \pRob{\AA}(s \xrightarrow{u} t) \geq \pmin^{2^h}.
\end{align}

The case $h = 0$ is for leaves.
Here, either $u$ is a letter $a$ and $\limu = \limu_+ = \bolda$,
or $u$ is the empty word $\varepsilon$ and $\limu = \limu_+ = \boldeps$.
Then both~\eqref{eq:induc1} and~\eqref{eq:induc2} hold.

Assume $h > 0$, there are two cases.

{\bf First case: a concatenation node} labelled
by $(u,(\limu,\limu_+))$ with two children 
labelled by $(u_1,(\limu_1,\limu_{+,1}))$
and $(u_2,(\limu_2,\limu_{+,2}))$.
By definition $u = u_1 \cdot u_2$, 
$\limu = \limu_1 \cdot \limu_2$ and 
$\limu_+ = \limu_{+,1} \cdot \limu_{+,2}$.

We first prove that~\eqref{eq:induc1} holds.
Indeed, for $s,t \in Q$, $\limu_+(s,t) = 1$ if and only if there exists $r \in Q$
such that $\limu_{+,1}(s,r) = 1$ and $\limu_{+,2}(r,t) = 1$.
On the other side, since:
\[
\pRob{\AA}(s \xrightarrow{u} t) = 
\sum_{r \in Q} \pRob{\AA}(s \xrightarrow{u_1} r) \cdot \pRob{\AA}(r \xrightarrow{u_2} t),
\]
then $\pRob{\AA}(s \xrightarrow{u} t) > 0$ if and only if
there exists $r \in Q$ such that $\pRob{\AA}(s \xrightarrow{u_1} r) \cdot \pRob{\AA}(r \xrightarrow{u_2} t) > 0$,
which is equivalent to $\pRob{\AA}(s \xrightarrow{u_1} r) > 0$ and $\pRob{\AA}(r \xrightarrow{u_2} t) > 0$.
We conclude with the induction hypothesis.

Now we prove that~\eqref{eq:induc2} holds.
Let $s,t \in Q$ such that $\limu(s,t) = 1$.
Then there exists $r \in Q$ such that
$\limu_1(s,r) = 1$ and $\limu_2(r,t) = 1$.
So:
\[
\pRob{\AA}(s \xrightarrow{u} t) \ge 
\pRob{\AA}(s \xrightarrow{u_1} r) \cdot \pRob{\AA}(r \xrightarrow{u_2} t) \ge 
\pmin^{2^h} \cdot \pmin^{2^h} = \pmin^{2^{h+1}},
\]
where the second inequality is by induction hypothesis.
This completes the proof of~\eqref{eq:induc2}.

{\bf Second case: an iteration node} labelled
by $(u,(\limu^\sharp,\limu_+))$ with $k$ sons labelled
by $(u_1,(\limu,\limu_{+})),\ldots,(u_k,(\limu,\limu_{+}))$.
By definition, $u = u_1 \cdots u_k$, and $(\limu,\limu_+)$ is idempotent.

The proof that~\eqref{eq:induc1} holds is similar to the concatenation node case.

Now we prove that~\eqref{eq:induc2} holds.
Let $s,t \in Q$ such that $\limu^\sharp(s,t) = 1$.
Since $k \geq 3$:
\begin{equation}\label{eq:induction_completeness_proof}
\pRob{\AA}(s \xrightarrow{u} t) \geq 
\pRob{\AA}(s \xrightarrow{u_1} t) \cdot 
\sum_{q \in Q} \pRob{\AA}(t \xrightarrow{u_2 \cdots u_{k-1}} q) \cdot 
\pRob{\AA}(q \xrightarrow{u_k} t).
\end{equation}

To establish~\eqref{eq:induc2} we prove that:
\begin{align}
\label{eq:induction_completeness_proof1}
& \pRob{\AA}(s \xrightarrow{u_1} t) \geq \pmin^{2^h},\\
\label{eq:induction_completeness_proof2}
& \text{for all } q \in Q,\ \pRob{\AA}(t \xrightarrow{u_2 \cdots u_{k-1}} q) > 0 \implies
\pRob{\AA}(q \xrightarrow{u_k} t) \geq \pmin^{2^h}.
\end{align}

We prove~\eqref{eq:induction_completeness_proof1}.
Since $\limu^\sharp(s,t) = 1$, by definition $\limu(s,t) = 1$ and $t$ is $\limu$-recurrent.
The induction hypothesis for the node $(u_1,(\limu,\limu_+))$ 
implies that $\pRob{\AA}(s \xrightarrow{u_1} t) \geq \pmin^{2^h}$,
\textit{i.e.}~\eqref{eq:induction_completeness_proof1}.

Now we prove~\eqref{eq:induction_completeness_proof2}.
For that we use the hypothesis that $(\limu,\limu_+)$ is not a leak witness.
Let $q \in Q$ such that $\pRob{\AA}(t \xrightarrow{u_2 \cdots u_{k-1}} q) > 0$.
By induction hypothesis for each child,~\eqref{eq:induc1} 
implies that $\limu_+^{k-2}(t,q) = 1$.
Since $\limu_+$ is idempotent, $\limu_+(t,q) = 1$. 
We argue that $\limu(q,t) = 1$. Let $\ell \in Q$ a $\limu$-recurrent state such that $\limu(q,\ell) = 1$.
Then $\limu_+(t,\ell) = 1$, and $t,\ell$ are $\limu$-recurrent. Since $(\limu,\limu_+)$ is not a leak witness,
it follows that $\limu(t,\ell) = 1$,
which implies that $\limu(\ell,t) = 1$ since $t$ is $\limu$-recurrent.
Together with $\limu(q,\ell) = 1$, this implies $\limu(q,t) = 1$.	
Thus, by induction hypothesis and according to~\eqref{eq:induc2},
$\pRob{\AA}(q \xrightarrow{u_k} t) \geq \pmin^{2^h}$,
so~\eqref{eq:induction_completeness_proof2} holds.

Now, putting~\eqref{eq:induction_completeness_proof},~\eqref{eq:induction_completeness_proof1}
and~\eqref{eq:induction_completeness_proof2} altogether:
\begin{align*}
\pRob{\AA}(s \xrightarrow{u} t) & \geq 
\pRob{\AA}(s \xrightarrow{u_1} t) \cdot 
\sum_{q \in Q} \pRob{\AA}(t \xrightarrow{u_2 \cdots u_{k-1}} q) \cdot 
\pRob{\AA}(q \xrightarrow{u_k} t) \\
& \geq \pmin^{2^h} \cdot \sum_{q \in Q} \pRob{\AA}(t \xrightarrow{u_2 \cdots u_{k-1}} q) \cdot \pmin^{2^h}\\
& = \pmin^{2^{h+1}},
\end{align*}
where the last equality holds because $\sum_{q\in Q} \pRob{\AA}(t \xrightarrow{u_2 \cdots u_{k-1}} q) = 1$.
This completes the proof of~\eqref{eq:induc2}.

To conclude, note that $\monoid_+$ has less than $3^{|Q|^2}$ elements.
\end{proof}

\subsection{Completeness of the Markov monoid algorithm for leaktight automata}
\label{subsec:completeness}

In this subsection we rely on the lower bound lemma (Lemma~\ref{lem:lowerbound})
to prove Theorem~\ref{theo:completeness}.
Let $\AA$ be a leaktight automaton.
By Lemma~\ref{lem:leaktight_characterization_direct}, its extended Markov monoid
does not contain any leak witness,
hence Lemma~\ref{lem:lowerbound} applies.

We prove the completeness of the Markov monoid associated with $\AA$.
Let $(u_n)\nNN$ be a sequence of finite words.
By Lemma~\ref{lem:lowerbound}, for each word $u_n$
there exists $(\limu_n,\limu_{+,n})$ in the extended Markov monoid such that
for all states $s,t$:
\[
\limu_n(s,t) = 1 \implies \pRob{\AA}(s \xrightarrow{u_n} t) \geq \dabound.
\]
Since the set of limit-words is finite, there exists $N \in \NN$
such that $\set{n \in \NN \mid \limu_N = \limu_n}$ is infinite.
To complete the proof, we prove that $\limu_N$ satisfies,
for all states $s,t$:
\[
\limsup \pRob{\AA}(s \xrightarrow{u_n} t) = 0 \implies \limu_N(s,t) = 0.
\]
Assume $\limsup \pRob{\AA}(s \xrightarrow{u_n} t) = 0$, 
then $\limsup \pRob{\AA}(s \xrightarrow{u_n} t) < \dabound$ 
for $n$ sufficiently large. 
Since $\limu_N = \limu_n$ for infinitely many $n \in \NN$,
this implies $\limu_N(s,t) = 0$,
which completes the proof of Theorem~\ref{theo:completeness}.

\section{Properties of leaktight automata}
\label{sec:leaktight_properties}
In this section, we extend the algorithm presented in Section~\ref{sec:algo},
and investigate its running complexity.
The extended algorithm has two features:
first, it checks \emph{at the same time} whether an automaton is leaktight
and whether it contains a value $1$ witness,
second, it runs in polynomial space.

We present an algebraic characterization of the leaktight property based on the extended Markov monoid,
allowing the extended algorithm to check the leaktight property.
For the complexity, one needs a deeper understanding of the Markov monoid;
in this section, we will show a linear bound on the $\sharp$-height,
allowing to compute the extended Markov monoid in polynomial space.
As a corollary, we obtain that the value $1$ problem
for leaktight automata is $\PSPACE$-complete.

\subsection{Characterization of the leaktight property}
In this subsection, we show the converse of Lemma~\ref{lem:leaktight_characterization_direct},
which implies the following theorem, characterizing the leaktight property in algebraic terms.

\begin{thm}
\label{theo:leaktight_characterization}
An automaton $\AA$ is leaktight if and only if its
extended Markov monoid does not contain any leak witness.
\end{thm}

Lemma~\ref{lem:lowerbound} is instrumental in the proof of this lemma.

\begin{proof}
We prove that if the extended Markov monoid of an automaton $\AA$
does not contain any leak witness,
then $\AA$ is leaktight.
The converse was proved in Lemma~\ref{lem:leaktight_characterization_direct}.

Assume $\AA$ has a leak $(u_n)\nNN$,
we show that its extended Markov monoid contains a leak witness.
Consider the Markov chain $\MM$ with state space $Q$ and transition matrix $M$ defined by $M(s,t) = \lim_n \pRob{\AA}(s \xrightarrow{u_n} t)$.
By assumption $M$ is idempotent.

By definition of a leak:
\begin{align}
\label{eq:proof_leaktight_converse_leak1}
r \text{ and } q \text{ are recurrent in } \MM, \\
\label{eq:proof_leaktight_converse_leak2}
M(r,q) = 0,\\
\label{eq:proof_leaktight_converse_leak3}
\text{for all } n \in \NN,\ \pRob{\AA}(r \xrightarrow{u_n} q) > 0.
\end{align}

Assume towards contradiction that the extended Markov monoid does not contain any leak witness, then Lemma~\ref{lem:lowerbound} applies.
For each word $u_n$, there exists $(\limu_n,\limu_{+,n})$ in the extended Markov monoid such that
for all states $s,t$:
\begin{align}
\label{eq:proof_leaktight_converse1}
\limu_{+,n}(s,t) = 1 & \iff \pRob{\AA}(s \xrightarrow{u_n} t) > 0,\\
\label{eq:proof_leaktight_converse2}
\limu_n(s,t) = 1 & \implies \pRob{\AA}(s \xrightarrow{u_n} t) \geq \dabound .
\end{align}
Since the extended Markov monoid is finite,
there exists $N \in \NN$ such that for infinitely many $n \in \NN$, we have $(\limu_N,\limu_{+,N}) = (\limu_n,\limu_{+,n})$.

Note that since each $u_n$ is idempotent,~\eqref{eq:proof_leaktight_converse1}
implies that each $\limu_{+,n}$ is idempotent as well.

Let $(\limv, \limv_+) = (\limu_N, \limu_{+,N})^{|Q|!}$.
The power $|Q|!$ ensures that $\limu_N^{|Q|!}$ is idempotent,
by Lemma~\ref{lem:basic_limit_words}.
Since $\limu_{+,N}$ is idempotent, $\limv_+ = \limu_{+,N}$.
Also, since $\limv$ is idempotent, there exists $r'$ and $q'$ which are $\limv$-recurrent,
such that $\limv(r,r') = 1$ and $\limv(q,q') = 1$, again thanks to Lemma~\ref{lem:basic_limit_words}.

Now, we prove that $(\limv,\limv_+)$ is a leak witness:
\begin{align}
\label{eq:proof_leaktight_converse_lk1}
r' \text{ and } q' \text{ are } \limv\text{-recurrent}, \\
\label{eq:proof_leaktight_converse_lk2}
\limv(r',q') = 0,\\
\label{eq:proof_leaktight_converse_lk3}
\limv_+(r',q') = 1.
\end{align}

Let $\eta = \dabound$ and $K = |Q|!$.

Observe that for all states $s,t$, we have $\limv(s,t) = 1 \implies M(s,t) > 0$:
\begin{align}
\notag & \limv(s,t) = 1 & \\
\notag & \implies \limu_N^K(s,t) = 1 & \text{(by definition of $\limv$)} \\
\notag & \implies \limu_n^K(s,t) = 1 \text{ for infinitely many $n$} & \text{(by definition of $N$)} \\
\notag & \implies \pRob{\AA}(s \xrightarrow{u_n^K} t) \geq \eta^K \text{ for infinitely many $n$} & \text{(by~\eqref{eq:proof_leaktight_converse2})} \\
\notag & \implies \lim_n \pRob{\AA}(s \xrightarrow{u_n^K} t) \geq \eta^K & \\
\notag & \implies M^K(s,t) > 0 & \text{(by definition of $M$)} \\
\notag & \implies M(s,t) > 0 & \text{(since $M$ is idempotent)}.
\end{align}

First,~\eqref{eq:proof_leaktight_converse_lk1} is by definition of $r'$ and $q'$.

We prove~\eqref{eq:proof_leaktight_converse_lk2}.
Towards contradiction, assume that $\limv(r',q') = 1$. 
Then $M(r',q') > 0$, so together with $M(r,r') > 0$ (which follows from $\limv(r,r') = 1$)
this implies $M^2(r,q') > 0$, so $M(r,q') > 0$ as $M$ is idempotent.
Since $M(q,q') > 0$ (which follows from $\limv(q,q') = 1$)
and $q$ is recurrent in $M$, we have $M(q',q) > 0$.
This implies $M^2(r,q) > 0$, and $M(r,q) > 0$ because $M$ is idempotent,
which contradicts~\eqref{eq:proof_leaktight_converse_leak2}.

We prove~\eqref{eq:proof_leaktight_converse_lk3}.
Thanks to~\eqref{eq:proof_leaktight_converse_leak3} and~\eqref{eq:proof_leaktight_converse1},
we have $\limu_{+,N}(r,q) = 1$, \textit{i.e.} $\limv_+(r,q) = 1$.
Since $M(r,r') > 0$ and $r$ is recurrent in $M$, we have $M(r',r) > 0$,
so~\eqref{eq:proof_leaktight_converse1} implies that $\limu_{+,N}(r',r) = 1$,
\textit{i.e.} $\limv_+(r',r) = 1$.
Similarly, $M(q,q') > 0$, so~\eqref{eq:proof_leaktight_converse1} implies that $\limu_{+,N}(q,q') = 1$,
\textit{i.e.} $\limv_+(q,q') = 1$.
The three equalities $\limv_+(r',r) = 1$, $\limv_+(r,q) = 1$ and $\limv_+(q,q') = 1$ imply $\limv_+^3(r',q') = 1$,
and since $\limv_+$ is idempotent $\limv_+(r',q') = 1$.

It follows that $(\limv,\limv_+)$ is a leak witness, which completes the proof.
\end{proof}

The immediate corollary of Theorem~\ref{theo:leaktight_characterization}
is that checking whether an automaton is leaktight can be done
by computing the extended Markov monoid and looking for leak witnesses,
hence it is decidable.

\subsection{The extended Markov monoid algorithm}
Algorithm~\ref{algo:extended_markov_monoid} computes the extended Markov monoid,
and looks for value $1$ witnesses, which in the extended Markov monoid
is an extended limit-word $(\limu,\limu_+)$ such that $\limu$ is a value $1$ witness (in the Markov monoid).
If there is a value $1$ witness, then the automaton has value $1$, 
even if it is not leaktight, thanks to Theorem~\ref{theo:consistency}.
Otherwise, the algorithm looks for a leak witness; if there is no leak witness,
then the automaton is leaktight thanks to Theorem~\ref{theo:leaktight_characterization}, 
and it does not have value $1$ thanks to Theorem~\ref{theo:completeness}.
In case there is a leak witness, the automaton is not leaktight, and nothing can be said.

\begin{algorithm}[h!t]
\caption{The extended Markov monoid algorithm.}
\label{algo:extended_markov_monoid}
\SetAlgoLined
\KwData{A probabilistic automaton.}
     
$\monoidext \gets \set{(\bolda,\bolda) \mid a\in A} \cup \set{(\boldeps,\boldeps)}$.

\Repeat{there is nothing to add}{
	\If{there is $(\limu,\limu_+),(\limv,\limv_+)\in \monoidext$ such that 
	$(\limu \cdot \limv,\limu_+ \cdot \limv_+) \notin \monoidext$}{
	add $(\limu \cdot \limv,\limu_+ \cdot \limv_+)$ to $\monoidext$
	}

	\If{there is $(\limu,\limu_+) \in \monoidext$ such that $(\limu,\limu_+)$ is idempotent 
	and $(\limu^\sharp,\limu_+) \notin \monoidext$}{
	add $(\limu^\sharp,\limu_+)$ to $\monoidext$
	}
}

\eIf{there is a value $1$ witness in $\monoidext$}{
	return true\;}{
	\eIf{there is no leak witness in $\monoidext$}{
		return false\;}{
		return fail: the automaton is not leaktight\;}{
	}
}
\end{algorithm}

\subsection{Parallel composition and \texorpdfstring{$\PSPACE$}{PSPACE}-hardness}

The objective of this subsection is to prove the $\PSPACE$-hardness
of the value $1$ problem for leaktight automata.
To this end, we give a reduction from the emptiness problem 
of $n$ deterministic automata.
To prove that the reduction indeed constructs leaktight automata,
we need to show that deterministic automata are leaktight,
and the closure under parallel composition. 

\begin{prop}\label{prop:deterministic}
Deterministic automata are leaktight.
\end{prop}

\begin{proof}
For all limit-words $\limu \in \set{\bolda \mid a \in A} \cup \set{\boldeps}$,
for all states $s$, there exists a unique state $t$ such that $\limu(s,t) = 1$.
In particular, each recurrence class is formed of only one state with a self-loop.
This property is preserved by concatenation, and implies that the iteration operation is trivial,
\textit{i.e.} $\limu^\sharp = \limu$.
Consequently, for all extended limit-words $(\limu,\limu_+)$
in the extended Markov monoid, we have $\limu = \limu_+$,
which implies that there are no leak witnesses.
\end{proof}

\begin{defi}[Parallel composition]
Consider two probabilistic automata, denoted $\AA = (Q^\AA,q_0^\AA,\Delta^\AA,F^\AA)$ and $\BB = (Q^\BB,q_0^\BB,\Delta^\BB,F^\BB)$.
We assume that $Q^\AA$ and $Q^\BB$ are disjoint.

The parallel composition of $\AA$ and $\BB$ is:
\[
\AA\ ||\ \BB\ =\ (Q^\AA \uplus Q^\BB\ ,\ \delta_0\ ,\ \Delta\ ,\ F^\AA \cup F^\BB),
\]
where $\delta_0 = \frac{1}{2} \cdot q_0^\AA + \frac{1}{2} \cdot q_0^\BB$, and:
$$\Delta(q,a) = 
\begin{cases}
\Delta_\AA(q,a) & \textrm{ if } q \in Q_\AA, \\
\Delta_\BB(q,a) & \textrm{ if } q \in Q_\BB. \\
\end{cases}$$
\end{defi}
By definition, for $u \in A^*$, we have 
$\pRob{\AA||\BB}(u) = \frac{1}{2} \cdot \pRob{\AA}(u) + \frac{1}{2} \cdot \pRob{\BB}(u)$.
Note that in this definition, we allowed an initial probability distribution
rather than only one initial state.
This could be avoided by adding a new initial state that leads to each previous initial state
with probability half, 
but we do it here for technical convenience in the proof of the following proposition.

\begin{prop}\label{prop:parallel_composition}
The leaktight property is stable by parallel composition.
\end{prop}

\begin{proof}
The extended Markov monoid $\monoidext^{\AA||\BB}$ of the parallel composition embeds into 
the direct product $\monoidext^\AA \times \monoidext^\BB$
of the extended Markov monoids of each automaton.

Note that for $(\limu,\limu_+) \in \monoidext^{\AA||\BB}$,
if $\limu(s,t) = 1$, then either $s,t \in Q^\AA$ or $s,t \in Q^\BB$,
and similarly for $\limu_+$.
Relying on this, we map $(\limu,\limu_+) \in \monoidext^{\AA||\BB}$ to
$\left((\limu,\limu_+)[\AA]\ ,\ (\limu,\limu_+)[\BB]\right)$,
where $(\limu,\limu_+)[\AA]$ is the restriction to $\AA$ and similarly for $\BB$.
An easy induction on $(\limu,\limu_+)$
shows that this map is an embedding into $\monoidext^\AA \times \monoidext^\BB$.

Consequently, the extended Markov monoid of the parallel composition contains a leak witness
if and only if one of the extended Markov monoid contains a leak witness.
\end{proof}

Now that we proved that deterministic automata are leaktight,
and the closure under parallel composition,
the $\PSPACE$-hardness of the value $1$ problem for leaktight automata
is easy.

\begin{prop}\label{prop:pspace_hard}
The value $1$ problem for leaktight automaton is $\PSPACE$-hard.
\end{prop}

\begin{proof}
We give a reduction from the following problem:
given $n$ deterministic automata over finite words, 
decide whether the intersection of the languages they accept is empty. 
This problem is $\PSPACE$-hard~\cite{K77}. 

The reduction is as follows: given $n$ deterministic automata, 
we construct the parallel composition of the $n$ automata,
where each copy is reached with probability $\frac{1}{n}$.
This automaton has value $1$ if and only if the intersection of the languages is not empty, 
and is leaktight by Proposition~\ref{prop:deterministic} and Proposition~\ref{prop:parallel_composition}.
\end{proof}

\subsection{Bounding the \texorpdfstring{$\sharp$}{sharp}-height in the Markov monoid}

We now consider the running complexity of the extended Markov monoid algorithm.
A na\"ive argument shows that it terminates in less than $3^{|Q|^2}$ iterations,
since each iteration adds a new extended limit-word in the monoid
and there are at most $3^{|Q|^2}$ different limit-words.
This gives an EXPTIME upper bound.

A better complexity can be achieved by looking for a value $1$ witness or a leak witness 
in a non-deterministic way.
The algorithm guesses the witness by its decomposition into concatenations and iterations.
The key observation, made by Kirsten~\cite{K05} in the context of distance desert automata,
is that the $\sharp$-height, that is the number of nested applications of the iteration operation,
can be restricted to at most $|Q|$.

Note that when dealing with $\sharp$-height, it suffices to consider limit-words
instead of extended limit-words,
as by definition the second component of an extended limit-word does not contain any $\sharp$.

Formally, we define the $\sharp$-hierarchy inside the Markov monoid as follows:
\begin{align}
\notag & S_0 = \langle \set{\bolda \mid a \in A} \cup \set{\boldeps} \rangle , \\
\notag & S_{p+1} = \langle S_p \cup \set{\limu^\sharp \mid \limu \in E(S_p)} \rangle ,
\end{align}
where $\langle T \rangle$ is the set of limit-words obtained as concatenation of limit-words in $T$.

\begin{defi}[$\sharp$-height of a limit-word]\label{def:sharp_height}
The $\sharp$-height of a limit-word $\limu$ is the minimal $p$ such that $\limu \in S_p$.
\end{defi}

\begin{thm}\label{theo:sharp_hierarchy}
Every limit-word has $\sharp$-height at most $|Q|$,
\textit{i.e.} the $\sharp$-hierarchy collapses at level $|Q|$.
\end{thm}

In the following, we adapt Kirsten's proof from~\cite{K05} to the setting of probabilistic automata.
Roughly speaking, the proof consists in associating a quantity to each idempotent element of the Markov monoid,
and to show the following:
\begin{itemize}
	\item the quantity is bounded above by $|Q|$.
	\item the quantity strictly decreases when iterating an unstable limit-word (\textit{i.e.} if $\limu^\sharp \neq \limu$),
	\item the quantity does not increase when concatenating.
\end{itemize}

Let $\limu$ be an idempotent limit-word, we define $\sim_\limu$ the relation on $Q$ 
by $s \sim_\limu t$ if $\limu(s,t) = 1$ and $\limu(t,s) = 1$.
Clearly, $\sim_\limu$ is symmetric, and since $\limu$ is idempotent, $\sim_\limu$ is transitive.
If for some state $s$ there exists a state $t$ such that $s \sim_\limu t$, then $s \sim_\limu s$ since $\limu$ is idempotent.
Consequently, the restriction of $\sim_\limu$ to the set
\[
Z_\limu = \set{s \in Q \mid s \sim_\limu s}
\]
is reflexive, \textit{i.e.} $\sim_\limu$ is an equivalence relation on $Z_\limu$.
From now on by equivalence class of $\sim_\limu$ we mean an equivalence class of $\sim_\limu$ on $Z_\limu$.
We denote by $[s]_\limu$ the equivalence class of $s$, 
and by $\Cl(\limu)$ the set of equivalence classes of $\sim_\limu$.
The quantity associated with $\limu$ is $|\Cl(\limu)|$, the number of equivalence classes of $\sim_\limu$,
that is the number of non-trivial connected components in the underlying graph of $\limu$.
Note that $|\Cl(\limu)| \le |Q|$.

Here are two useful observations.

\begin{lem}\label{lem:sharp_height_tool}\hfill
\begin{itemize}
	\item Let $\limu,\limv$ be two limit-words and $s,t,r \in Q$.
Then $(\limu \cdot \limv)(s,t) \geq \limu(s,r) \cdot \limv(r,t)$.
	\item Let $\limu$ be an idempotent limit-word and $s,t \in Q$.
There exists $r \in Q$ such that $\limu(s,t) = \limu(s,r) \cdot \limu(r,r) \cdot \limu(r,t)$.
\end{itemize}
\end{lem}

\begin{proof}
The first claim is clear and follows from the equality:
\[
(\limu \cdot \limv)(s,t) = \sum_{r \in Q} \limu(s,r) \cdot \limv(r,t).
\]
Consider now the second claim. For all states $r \in Q$, since $\limu$ is idempotent we have:
\[
\limu(s,t) = \limu^3(s,t) = \sum_{p,q \in Q} \limu(s,p) \cdot \limu(p,q) \cdot \limu(q,t) 
\geq \limu(s,r) \cdot \limu(r,r) \cdot \limu(r,t).
\]
Since $\limu$ is idempotent, we have $\limu = \limu^{n+2}$, so there exist $s = r_0,\ldots,r_{n+2} = t$ 
such that $\limu(s,t) = \limu(r_0,r_1) \cdots \limu(r_{n+1},r_{n+2})$. 
By a counting argument, there exist $i,j$ such that $1 \leq i < j \leq (n + 1)$ and $r_i = r_j$, denote it by $r$.
We have:
\begin{align}
\notag & \limu(s,r) = \limu^i(s,r) \geq \limu(r_0,r_1) \cdots \limu(r_{i-1},r_i), \\
\notag & \limu(r,r) = \limu^{j-i}(r,r) \geq \limu(r_i,r_{i+1}) \cdots \limu(r_{j-1},r_j),\\
\notag & \limu(r,t) = \limu^{n+2-j}(r,t) \geq \limu(r_j,r_{j+1}) \cdots \limu(r_{n+1},r_{n+2}).
\end{align}
Hence, $\limu(s,r) \cdot \limu(r,r) \cdot \limu(r,t) \geq \limu(r_0,r_1) \cdots \limu(r_{n+1},r_{n+2}) = \limu(s,t)$,
and the second claim follows.
\end{proof}

The following lemma shows that the quantity $|\Cl(\limu)|$ 
strictly decreases when iterating an unstable limit-word (\textit{i.e.} if $\limu^\sharp \neq \limu$).

\begin{lem}\label{lem:sharp_height_decrease}
Let $\limu$ be an idempotent limit-word. 
\begin{align}
\label{eq:sharp_height1}
\Cl(\limu^\sharp) \subseteq \Cl(\limu), \\
\label{eq:sharp_height2}
\text{if } \limu \neq \limu^\sharp, \text{ then } \Cl(\limu^\sharp) \neq \Cl(\limu).
\end{align}
\end{lem}

\begin{proof}
We prove~\eqref{eq:sharp_height1}.
Let $s \in Z_{\limu^\sharp}$; by definition we have $s \sim_{\limu^\sharp} s$.
We show that $[s]_{\limu^\sharp} = [s]_\limu$.
For all states $t \in Q$ such that $s \sim_{\limu^\sharp} t$, we have $s \sim_\limu t$, so $[s]_{\limu^\sharp} \subseteq [s]_\limu$.
Conversely, let $t \in [s]_\limu$; we have $\limu(s,t) = 1$ and $\limu(t,s) = 1$. 
Since $s \sim_{\limu^\sharp} s$, we have $\limu^\sharp(s,s) = 1$. 
So $\limu^\sharp(s,t) = (\limu^\sharp \cdot \limu)(s,t) \geq \limu^\sharp(s,s) \cdot \limu(s,t) = 1$,
and similarly $\limu^\sharp(t,s) = (\limu \cdot \limu^\sharp)(t,s) \geq \limu(t,s) \cdot \limu^\sharp(s,s) = 1$.
Thus $s \sim_{\limu^\sharp} t$, \textit{i.e.} $t \in [s]_{\limu^\sharp}$,
which concludes to the equality $[s]_{\limu^\sharp} = [s]_\limu$.
In other words, the equivalence classes for $\limu^\sharp$ are also equivalence classes for $\limu$,
so $\Cl(\limu^\sharp) \subseteq \Cl(\limu)$.

We prove~\eqref{eq:sharp_height2}.
Assume $\limu \neq \limu^\sharp$; let $s,t$ such that $\limu(s,t) = 1$ and $\limu^\sharp(s,t) = 0$.
By Lemma~\ref{lem:sharp_height_tool}, 
there exists $r$ such that $\limu(s,t) = \limu(s,r) \cdot \limu(r,r) \cdot \limu(r,t)$,
so $\limu(r,r) = 1$.
Towards contradiction, assume $\limu^\sharp(r,r) = 1$. 
It follows that:
\[
\limu^\sharp(s,t) = (\limu \cdot \limu^\sharp \cdot \limu)(s,t) \geq 
\limu(s,r) \cdot \limu^\sharp(r,r) \cdot \limu(r,t) = 
\limu(s,r) \cdot \limu(r,r) \cdot \limu(r,t) = 1,
\]
\textit{i.e.}, $\limu^\sharp(s,t) = 1$ which is a contradiction. 
Consequently, $\limu^\sharp(r,r) = 0$, so $r \sim_\limu r$ and $r \not\sim_{\limu^\sharp} r$. 
Thus, $r \in Z_\limu$ but $r \notin Z_{\limu^\sharp}$. 
Hence, there is a class $[r]_\limu$ in $\Cl(\limu)$, 
but there is no class $[r]_{\limu^\sharp}$ in $\Cl(\limu^\sharp)$. 
\end{proof}

The following lemma shows that the quantity $|\Cl(\limu)|$ is common to 
all idempotents in the same $\JJ$-class.
The notion of $\JJ$-class is a classical notion for the theory of monoids,
derived from one of the four Green's relations called the $\JJ$-preorder
(for details about Green's relations, see~\cite{L79,H95},
or~\cite{C11} for its applications to automata theory).

Define $\limu \le_\JJ \limv$ if there exist $\lima,\limb$
such that $\lima \cdot \limv \cdot \limb = \limu$,
and $\limu \JJ \limv$, \textit{i.e.} $\limu$ and $\limv$ are in the same $\JJ$-class,
if $\limu \le_\JJ \limv$ and $\limv \le_\JJ \limu$.

\begin{lem}\label{lem:sharp_height_concatenation}
Let $\limu,\limv$ be idempotent limit-words. 
If $\limu \leq_\JJ \limv$, then $|\Cl(\limu)| \leq |\Cl(\limv)|$.
\end{lem}

\begin{proof}
Let $\lima,\limb$ two limit-words such that $\lima \cdot \limv \cdot \limb = \limu$. 
First, without loss of generality we assume that $\lima \cdot \limv = \lima$ and $\limv \cdot \limb = \limb$.
Indeed, if $\lima$ and $\limb$ do not satisfy these conditions, 
then we consider $\lima = \lima \cdot \limv$ and $\limb = \limv \cdot \limb$.

We construct a partial surjective mapping $\beta : \Cl(\limv) \rightarrow \Cl(\limu)$,
which depends on the choice of $\lima$ and $\limb$. 
For all states $s \in Z_\limv$ and $t \in Z_\limu$ satisfying $\lima(t,s) \cdot \limv(s,s) \cdot \limb(s,t) = 1$
we set $\beta([s]_\limv) = [t]_\limu$. 
To complete the proof, we have to show that $\beta$ is well defined and that $\beta$ is indeed surjective.

We show that $\beta$ is well defined. 
Let $s,s' \in Z_\limv$ and $t,t' \in Z_\limu$, and assume
$\lima(t,s) \cdot \limv(s,s) \cdot \limb(s,t) = 1$
and $\lima(t',s') \cdot \limv(s',s') \cdot \limb(s',t') = 1$.
By definition $\beta([s]_\limv) = [t]_\limu$ and $\beta([s']_\limv) = [t']_\limu$. 
To show that $\beta$ is well defined, we have to show that if
$[s]_\limv = [s']_\limv$, then $[t]_\limu = [t']_\limu$. 
Assume $[s]_\limv = [s']_\limv$, \textit{i.e.}, $s \sim_\limv s'$, so $\limv(s,s') = 1$.
Since $\lima(t,s) \cdot \limv(s,s) \cdot \limb(s,t) = 1$, we have
$\lima(t,s) = \limb(s,t) = 1$. 
Similarly, $\lima(t',s') \cdot \limv(s',s') \cdot \limb(s',t') = 1$, 
so $\lima(t',s') = \limb(s',t') = 1$.
Consequently, $\lima(t,s) \cdot \limv(s,s') \cdot \limb(s',t') = 1$, so 
$\limu(t,t') = (\lima \cdot \limv \cdot \limb)(t,t') = 1$. 
Symmetrically, $\lima(t',s') \cdot \limv(s',s) \cdot \limb(s,t) = 1$, so 
$\limu(t',t) = 1$, concluding to $t \sim_\limu t'$, \textit{i.e.} $[t]_\limu = [t']_\limu$.

We show that $\beta$ is surjective. 
Let $t \in Z_\limu$.
We exhibit some $s$ such that $\beta([s]_\limv) = [t]_\limu$. 
Since $\limu = \lima \cdot \limv \cdot \limb$, there are $p, q$ such that
$\lima(t,p) \cdot \limv(p,q) \cdot \limb(q,t) = \limu(t,t) = 1$, 
so $\lima(t,p) = \limv(p,q) = \limb(q,t) = 1$. 
By Lemma~\ref{lem:sharp_height_tool} there exists $s$ such
that $\limv(p,s) \cdot \limv(s,s) \cdot \limv(s,q) = \limv(p,q) = 1$, 
so, $\limv(p,s) = \limv(s,s) = \limv(s,q) = 1$. 
We have $\lima(t,s) = (\lima \cdot \limv)(t,s) \geq \lima(t,p) \cdot \limv(p,s) = 1$, 
and $\limb(s,t) = (\limv \cdot \limb)(s,t) \geq \limv(s,q) \cdot \limb(q,t) = 1$. 
To sum up, $\lima(t,s) \cdot \limv(s,s) \cdot \limb(s,t) = 1$,
and hence, $\beta([s]_\limv) = [t]_\limu$.
\end{proof}

The following lemma wraps up the previous two lemma.
For technical convenience, we set $S_{-1} = \emptyset$.

\begin{lem}\label{lem:sharp_height_conclusion}
Let $\limu$ be an idempotent limit-word and $p \ge 0$.
If $\limu \in S_p \setminus S_{p-1}$, then $|\Cl(\limu)| \le |Q| - p$.
\end{lem}

\begin{proof}
We proceed by induction on $p$. 
For $p = 0$, the assertion is obvious.
Let $p \ge 0$, we show the claim for $p + 1$.
Let $\limu$ be an idempotent limit-word such that $\limu \in S_{p+1} \setminus S_p$.
By definition, $\limu = \limv_1 \cdots \limv_k$ where for all $i$, 
either $\limv_i \in S_p$ or $\limv_i = \limu_i^\sharp$ for $\limu_i \in S_p$ and $\limv_i \notin S_p$.

If for all $i$, $\limv_i \in S_p$,
then $\limu = \limv_1 \cdots \limv_k \in S_p$, which is a contradiction. 
Consequently, there exists $i$ such that $\limv_i = \limu_i^\sharp$ for $\limu_i \in S_p$ and $\limv_i \notin S_p$.
Since $\limu_i \in S_p$ and $\limv_i = \limu_i^\sharp \notin S_p$, we have $\limu_i^\sharp \neq \limu_i$. 
Towards contradiction, assume $\limu_i \in S_{p-1}$, 
then $p \ge 1$, and this implies $\limu_i^\sharp \in S_p$, which is a contradiction. 
Hence, $\limu_i \in S_p \setminus S_{p-1}$.

By induction, we have $|\Cl(\limu_i)| \le |Q| - p$. 
Since $\limu_i^\sharp \neq \limu_i$, by Lemma~\ref{lem:sharp_height_decrease}
we have $|\Cl(\limu_i^\sharp)| < |\Cl(\limu_i)|$.
Since $\limu \le_\JJ \limu_i^\sharp$, by Lemma~\ref{lem:sharp_height_concatenation}
we have $|\Cl(\limu)| \le |\Cl(\limu_i^\sharp)|$.
Altogether, it follows $|\Cl(\limu)| \le |Q| - (p+1)$.
\end{proof}

It follows from Lemma~\ref{lem:sharp_height_conclusion}
that $S_{|Q|} = S_{|Q| + 1}$, \textit{i.e.} the $\sharp$-hierarchy collapses at level $|Q|$,
proving Theorem~\ref{theo:sharp_hierarchy}.

The bound is almost tight, as shown in figure~\ref{fig:sharpstar}.
The only value $1$ witness of this automaton is
$(\cdots((\bolda_0^\sharp\ \bolda_1)^\sharp\ \bolda_2)^\sharp\ \bolda_3)^\sharp\ \cdots \bolda_{n-1})^\sharp$,
whose $\sharp$-height is $|Q|-2$.
\begin{figure}[ht]
\begin{center}
\includegraphics[scale=.8]{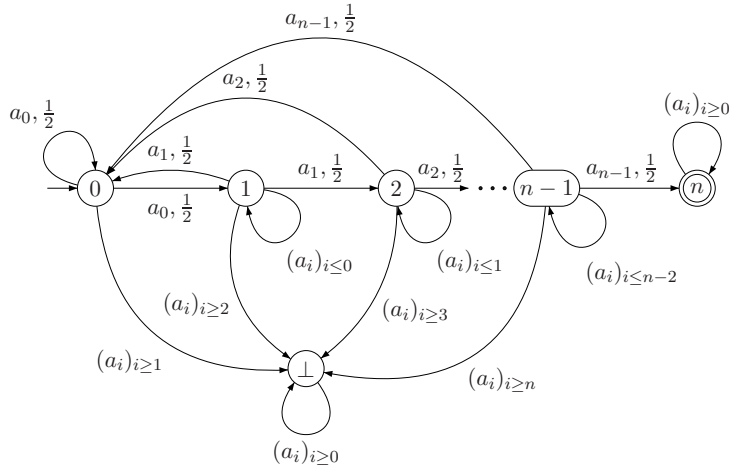}
\caption{\label{fig:sharpstar} A leaktight automaton with value $1$ and $\sharp$-height $|Q|-2$.}
\end{center}
\end{figure}
Note that this automaton is leaktight, so the extended Markov monoid algorithm will find the value $1$ witness
and correctly answers that it has value $1$.

\subsection{Finding witnesses in the Markov monoid}

In the subsection, we will prove the following complexity result.

\begin{prop}\label{prop:pspace}
There exists an algorithm which checks in polynomial space
whether an automaton is leaktight and whether in such case 
it has value $1$.
\end{prop}

Following Theorem~\ref{theo:completeness}, checking whether a leaktight automaton
has value $1$ boils down to finding a value $1$ witness in the Markov monoid.
Similarly, following Theorem~\ref{theo:leaktight_characterization},
checking whether an automaton is \emph{not} leaktight boils down to finding a leak witness
in the extended Markov monoid.
Note that in both cases, checking whether a given limit-word or extended limit-word
is a witness is easily done in polynomial time.

Since we aim at proving that those two tasks can be computed in $\PSPACE$,
which is closed under complementation,
it suffices to show how to find a witness in the (extended) Markov monoid.
For the sake of readability, we here only deal with the Markov monoid,
but similar ideas apply to the extended Markov monoid.

We describe an algorithm to guess a witness in the Markov monoid.
The key property given by Theorem~\ref{theo:sharp_hierarchy}
is that we can restrict ourselves to at most $|Q|$ nested iteration operations.

As the corresponding property was proved by Kirsten~\cite{K05} in the context of distance automata,
also to obtain a $\PSPACE$ algorithm, 
the following algorithm is also an adaptation of~\cite{K05}.
Rather than a formal proof, we here give an intuitive description of the algorithm.

A witness can be described as a tree whose nodes are labelled by limit-words, 
of depth at most $2 \cdot |Q| + 1$, as follows:
\begin{itemize}
	\item a leaf is labelled either by $\bolda$ for $a \in A$ or by $\boldeps$,
	\item an internal node can be a \emph{concatenation node}, then it is labelled
	by $\limu = \limv_1 \cdots \limv_k$ for $k \le 2^{|Q|^2}$ and has $k$ children,
	labelled by $\limv_1,\ldots,\limv_k$,
	\item an internal node can be an \emph{iteration node}, then it is labelled
	by $\limu^\sharp$ and has one child labelled $\limu$.
\end{itemize}
We describe an algorithm that guesses such a tree.
It starts from the root, and travels over nodes in a depth-first way: from top to bottom (and up again) 
and from left to right.
In a node, the algorithm stores the branch that leads to this node,
and for each node in the branch the limit-word obtained by concatenating
all the left siblings of this node.
From a node, the algorithm guesses a limit-word,
and whether it will be a leaf, a concatenation node or an iteration node.
In the first case, it goes up and checks the consistency of this guess.
In the two other cases, it updates the value of this node
by concatenating the new guess with the previous value
and goes down.

Although the tree is of exponential size, 
in each step the algorithm only stores $2 \cdot |Q| + 1$ limit-words at most,
so it runs in polynomial space.

\section{Examples and subclasses of leaktight automata}
\label{sec:leaktight_comparisons}

In this section, we investigate further the class of leaktight automata,
by giving examples of leaktight automata, exhibiting subclasses, and showing closure properties.
In particular, we prove that hierarchical automata, $\sharp$-acylic automata and simple automata
are all strict subclasses of leaktight automata.
(Actually, since $\sharp$-acylic automata are already a subclass of simple automata, we do not consider them.)
This implies that our decidability result extends the decidability results from~\cite{GO10,CT12}.

\subsection{Two basic examples}
The automaton on figure~\ref{fig:3} is leaktight.
As we shall see, it is not hierarchical, nor simple, hence it witnesses
that leaktight automata are not subsumed by hierarchical or simple automata.
Its extended Markov monoid is depicted on the right-hand side.
Each of the four directed graphs represents an extended limit-word
$(\limu,\limu_+)$: if $\limu(s,t) = 1$, then $(s,t)$ is an edge,
and if $\limu(s,t) = 0$ but $\limu_+(s,t) = 1$, then $(s,t)$ is marked with $+$.

The initial state of the automaton is state $0$, and the unique final state is state $1$.
This automaton has value $1$ and this can be checked using the
extended Markov monoid: the two value $1$ witnesses are $\lima^\sharp$ and 
$\limb \cdot \lima^\sharp$.

\begin{figure}[ht]
\begin{center}
\includegraphics[scale=1]{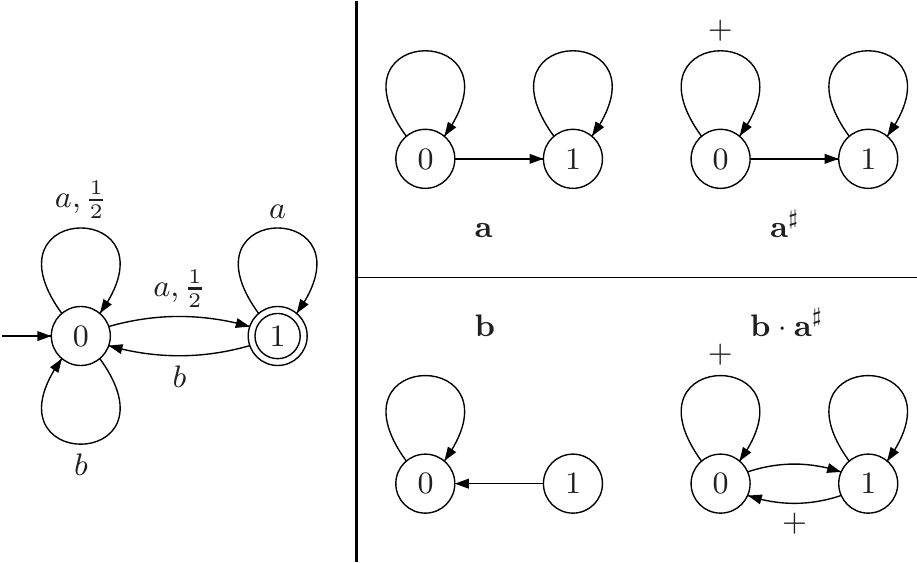}
\caption{\label{fig:3} A leaktight automaton and its extended Markov monoid.}
\end{center}
\end{figure}

The automaton on figure~\ref{fig:4} is leaktight.
The initial state of the automaton is state $0$, and the unique final state is state $F$.
The Markov monoid has too many elements to be represented here.
This automaton does not have value $1$.

\begin{figure}[ht]
\begin{center}
\includegraphics[scale=1]{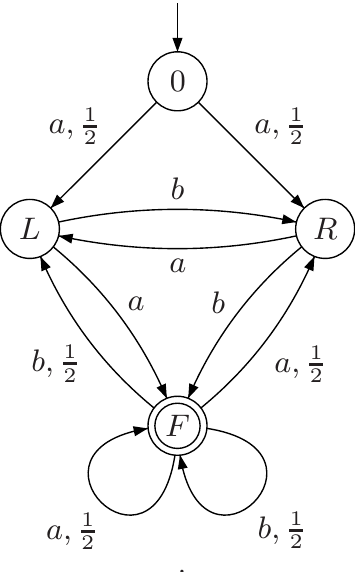}
\caption{\label{fig:4} A leaktight automaton which does not have value $1$.}
\end{center}
\end{figure}

\subsection{Some closure properties}
In this subsection, we show some closure properties: synchronised product and 
composition with a deterministic transducer.
As for Proposition~\ref{prop:parallel_composition},
the automata constructions reflect in algebraic constructions,
allowing to use the characterization 
with leak witnesses given by Theorem~\ref{theo:leaktight_characterization}.

In this subsection, we will omit initial and final states,
as we are only interested in preserving the leaktight property,
which does not depend on it.

\begin{defi}[Synchronised product]
Let $\AA = (Q^\AA,\Delta^\AA)$ and $\BB = (Q^\BB,\Delta^\BB)$ 
be two probabilistic automata.

A synchronised product of $\AA$ and $\BB$ is:
\[
\AA \times \BB\ =\ (Q^\AA \times Q^\BB\ ,\ \Delta),
\]
where $\Delta(q,a) = (\Delta_\AA(q,a),\Delta_\BB(q,a))$.
\end{defi}

\begin{prop}
The leaktight property is stable by synchronized product.
\end{prop}

\begin{proof}
The extended Markov monoid $\monoidext^{\AA \times \BB}$ of the synchronized product
embeds into the direct product $\monoidext^\AA \times \monoidext^\BB$ of the extended Markov monoids of each automaton.

Let $(\limu,\limu_+)$ be an extended limit-word in $\monoidext^{\AA \times \BB}$.
Define $\limu_\AA(s,t) = 1$ if there exists $s',t' \in Q^\BB$ such that 
$\limu((s,s'),(t,t')) = 1$, and similarly for $\limu_{+,\AA}$, $\limu_\BB$ and $\limu_{+,\BB}$.
We have the following equivalence:
\[
\limu((s,s'),(t,t')) = 1 \iff \limu_\AA(s,t) = 1 \wedge \limu_\BB(s',t') = 1,
\]
and similarly for $\limu_+$.

Relying on this, we map $(\limu,\limu_+) \in \monoidext^{\AA \times \BB}$ to
$\left((\limu_\AA,\limu_{+,\AA})\ ,\ (\limu_\BB,\limu_{+,\BB})\right)$.
An easy induction shows that this map is an embedding into $\monoidext^\AA \times \monoidext^\BB$.

Consequently, the extended Markov monoid of the synchronised product contains a leak witness
if and only if one of the extended Markov monoid contains a leak witness.
\end{proof}

The last closure property we prove will be useful in the next section.
\begin{defi}[Composition with a deterministic transducer]
Let $\AA = (Q^\AA,\Delta^\AA)$ be a probabilistic automaton,
and $\MM = (Q^\MM,Q^\AA,\Delta^\MM)$ a deterministic transducer over $\AA$,
\textit{i.e.} $\Delta^\MM : Q^\MM \times Q^\AA \rightarrow Q^\MM$.

The composition of $\AA$ by $\MM$ is:
$$\AA \otimes \MM\ =\ (Q^\AA \times Q^\MM\ ,\ \Delta),$$
where $\Delta(q,p,a) = (\Delta_\AA(q,a),\Delta_\MM(p,q))$.
\end{defi}

\begin{prop}\label{prop:deterministic_transducer}
The leaktight property is stable by composition with a deterministic transducer.
\end{prop}

\begin{proof}
Following the same reasoning as in~\ref{prop:deterministic},
one can show that the extended Markov monoids for $\AA$ and for $\AA \otimes \MM$ are isomorphic.
\end{proof}

\subsection{Leaktight automata strictly contain hierarchical automata}

The class of hierarchical automata has been defined in~\cite{CSV11}, where it was proved
that they are recognize exactly the class of $\omega$-regular languages.
The states $Q$ of a hierarchical automaton are sorted according to levels
such that for each letter, at most one successor is at the same level and all
others are at higher levels.

Formally, there exists a mapping $\rank : Q \to \NN$
such that for all $a \in A$, for all states $s,t$ such that $\pRob{\AA}(s \xrightarrow{a} t) > 0$,
we have $\rank(s) \leq \rank(t)$.
Furthermore, if $\pRob{\AA}(s \xrightarrow{a} t) > 0$ and 
$\pRob{\AA}(s \xrightarrow{a} t') > 0$ but $\rank(s) = \rank(t) = \rank(t')$,
then $t = t'$.

\begin{prop}\label{prop:hierarchical}
Every hierarchical automata is leaktight.
\end{prop}

\begin{proof}
We prove by induction that for every extended limit-word $(\limu,\limu_+)$
in the extended Markov monoid of a hierarchical automaton,
for every states $s,t,t'$:
\begin{align}
\label{eq:hierarchical1}
\limu_+(s,t) = 1 \implies \rank(s) \le \rank(t),\\
\label{eq:hierarchical2}
\limu_+(s,t) = 1\ \wedge\ \limu_+(s,t') = 1\ \wedge\ \rank(s) = \rank(t) = \rank(t')\ \implies \ t = t',\\
\label{eq:hierarchical3}
\limu(s,t) = 1 \implies \rank(s) \le \rank(t),\\
\label{eq:hierarchical4}
\limu(s,t) = 1\ \wedge\ \limu(s,t') = 1\ \wedge\ \rank(s) = \rank(t) = \rank(t')\ \implies \ t = t'.
\end{align}
Note that~\eqref{eq:hierarchical1} and~\eqref{eq:hierarchical2}
imply~\eqref{eq:hierarchical3} and~\eqref{eq:hierarchical4},
since $\limu(s,t) = 1$ implies $\limu_+(s,t) = 1$.
The key property following from~\eqref{eq:hierarchical2} and~\eqref{eq:hierarchical3} 
is that the recurrence classes of $\limu$ and of $\limu_+$ are singletons.

This is trivial for $(\boldeps,\boldeps)$. 
The case of $(\bolda,\bolda)$ is the definition of hierarchical automata.
The induction step for concatenation is routinely checked, and trivial
for the iteration.

We now prove that the extended Markov monoid of a hierarchical automaton
does not contain any leak witness. We prove a slightly stronger statement; 
for every extended limit-word $(\limu,\limu_+)$
in the extended Markov monoid of a hierarchical automaton,
for all states $q,r$:
\begin{align}
\label{eq:hierarchical5}
\limu_+(r,q) = 1\ \wedge\ r \text{ is } \limu\text{-recurrent} \implies \ \limu(r,q) = 1.
\end{align}
This is clear for $(\boldeps,\boldeps)$ and for $(\bolda,\bolda)$.

Concatenation. Let $(\limu,\limu_+)$ and $(\limv,\limv_+)$
be two extended limit-words satisfying~\eqref{eq:hierarchical5}.
Consider two states $r,q$ such that $(\limu_+ \cdot \limv_+)(r,q) = 1$
and $r$ is $(\limu \cdot \limv)$-recurrent.

Since $r$ is $(\limu \cdot \limv)$-recurrent, the recurrence class of $r$ for $(\limu \cdot \limv)$
is $r$ itself, so $(\limu \cdot \limv)(r,r) = 1$
and $(\limu \cdot \limv)(r,p) = 1$ implies $p = r$.
Let $t$ be a state such that $\limu(r,t) = 1$ and $\limv(t,r) = 1$.
Thanks to~\eqref{eq:hierarchical3}, $\rank(t) = \rank(r)$.

We argue that $\limu(r,p) = 1$ implies $p = t$.
Indeed, let $p$ be a state such that $\limu(r,p) = 1$.
There exists a state $p'$ such that $\limv(p,p') = 1$,
so in particular $(\limu \cdot \limv)(r,p') = 1$, so $p' = r$.
By~\eqref{eq:hierarchical3}, $\rank(r) \le \rank(p) \le \rank(r)$,
so they are equal.
By~\eqref{eq:hierarchical4}, since $\limu(r,t) = 1$, $\limu(r,p) = 1$
and $\rank(r) = \rank(t) = \rank(p)$, we have $p = t$.
It follows that the state $t$ is $\limu$-recurrent:
by Lemma~\ref{lem:basic_limit_words}, there exists a state $p$
such that $\limu(r,p) = 1$ and $p$ is $\limu$-recurrent.
The above remark implies that $p = t$.

We argue that $\limv(t,p) = 1$ implies $p = r$.
Indeed, let $p$ be a state such that $\limv(t,p) = 1$.
We have $(\limu \cdot \limv)(r,p) = 1$, so $p = r$.
It follows that the state $r$ is $\limv$-recurrent:
by Lemma~\ref{lem:basic_limit_words}, there exists a state $p$
such that $\limv(t,p) = 1$ and $p$ is $\limu$-recurrent.
The above remark implies that $p = r$.

Since $(\limu_+ \cdot \limv_+)(r,q) = 1$, there exists $s \in Q$
such that $\limu_+(r,s) = 1$ and $\limv_+(s,q) = 1$.
By induction hypothesis for $(\limu,\limu_+)$,
since $\limu_+(r,s) = 1$ and $r$ is $\limu$-recurrent,
we have $\limu(r,s) = 1$.
Since $r$ is $\limu$-recurrent, its recurrence class is $r$ itself, so $s = r$.
By induction hypothesis for $(\limv,\limv_+)$,
since $\limv_+(r,q) = 1$ and $r$ is $\limv$-recurrent,
we have $\limv(r,q) = 1$.

It follows that $(\limu \cdot \limv)(r,q) = 1$.

The case of iteration is easy.
\end{proof}

The inclusion is strict, an example is given by figure~\ref{fig:3}.

\subsection{Leaktight automata strictly contain simple automata}
\label{subsec:simple_automata}

The class of simple automata has been defined in~\cite{CT12}, where
it was proved that the value $1$ problem is decidable for a subset of this class (namely, for structurally simple automata). 
In the following we show that the class of simple automata is strictly contained in the class of leaktight automata.

We fix $\AA$ a probabilistic automaton.

\begin{defi}[Jets] A jet is a sequence $(J_n)\nNN$ where $J_k\subseteq Q$ for each $k\in \NN$.
\end{defi}

\begin{defi}[Simple process~\cite{CT12}]\label{def:simple_process}
Let $w \in A^\omega$ be an infinite word.
The process induced by $w$ from the state $p$ is simple if there exists $\lambda > 0$ and two jets $(A_n)\nNN$, $(B_n)\nNN$
such that:
\begin{enumerate}
	\item for all $k\in \NN$, $A_k$ and $B_k$ are disjoint and $A_k \cup B_k = Q$,
	\item for all $k \in \NN$ and all $q \in A_k$, 
    $\pRob{\AA}(p \xrightarrow{w_{< k}} q) \ge \lambda$,
	\item $\lim_{n \to \infty} \pRob{\AA} (p \xrightarrow{w_{< n}} B_n) = 0$.
\end{enumerate}
\end{defi}

\begin{defi}[Simple automata~\cite{CT12}]\label{def:simple}
$\AA$ is simple if for every infinite word $w$ and every state $p$ the process induced by $w$ from $p$ is simple.
\end{defi}

\begin{thm}\label{theo:simple_inclusion_leaktight}
Every simple automaton is leaktight.
\end{thm}

The remainder of this subsection is devoted to the proof of Theorem~\ref{theo:simple_inclusion_leaktight}.
The proof is divided into two parts: first, we define non-simplicity witnesses, which are elements of the Markov monoid
that witnesses the non-simplicity of an automaton, and second we show that if the Markov monoid of an automaton contains a leak, 
then it also contains a non-simplicity witness.

\subsubsection{Non-simplicity witness}

\begin{defi}[Non-simplicity witness]\label{def:nonsimp} 
A triple $(\limu,\limv,\limw)$ of elements of the Markov monoid 
is a \emph{non-simplicity} witness if there exist states $r,t$ such that:
\begin{enumerate}
	\item $\limu \limv^\sharp \limw$ is idempotent,
	\item r is $\limu \limv^\sharp \limw$-recurrent,
	\item $\limu \limv(r,t) = 1$,
	\item t is $\limv$-transient.
\end{enumerate}
\end{defi}

\begin{prop}\label{prop:non_simplicity_witness_implies_not_simple}
If the Markov monoid of a probabilistic automaton contains a non-simplicity witness,
then it is not simple.
\end{prop}

To prove Proposition~\ref{prop:non_simplicity_witness_implies_not_simple}, we rely on the following two lemmata:

\begin{lem}\label{lem:rec} 
Let $\limu$ be an idempotent element of the Markov monoid, $r$ be a state $\limu$-recurrent, 
and $(u_n)\nNN$ a sequence of words that reifies $\limu$.
Then there exists a constant $\gamma > 0$ and a strictly increasing map $h : \NN \to \NN$ 
such that for all $n \in \NN$, we have:
\[
\pRob{\AA}(r \xrightarrow{u_{h(0)} \cdots u_{h(n-1)}} r) \ge \gamma.
\]
\end{lem}

\begin{proof}
Since the sequence of words $(u_n)\nNN$ reifies the limit-word $\limu$,
for all states $s,t$, $\pRob{\AA}(s \xrightarrow{u_n} t)$ converges and:
\begin{equation}\label{eq:lim}
\limu(s,t) = 1 \iff \lim_n \pRob{\AA}(s \xrightarrow{u_n} t) > 0.
\end{equation}

Define:
$$\lambda = \frac{1}{2} \cdot \min \set{\lim_n \pRob{\AA}(s \xrightarrow{u_n} t) \mid \limu(s,t) = 1} .$$
Thanks to~\eqref{eq:lim}, there exists an increasing map $h : \NN \to \NN$ such that 
the following two conditions hold, for all states $s,t$ and for all $n \in \NN$:
\begin{equation}\label{eq:ln} 
	\text{ if } \limu(s,t) = 0, \text{ then }
	\pRob{\AA}(s \xrightarrow{u_{h(n)}} t) \le \frac{1}{|Q| \cdot 2^{n+2}} ,
\end{equation}  
\begin{equation} \label{eq:bound}
	\text{ if } \limu(s,t) = 1, \text{ then }
	\pRob{\AA}(s \xrightarrow{u_{h(n)}} t) \ge \lambda .
\end{equation}

We now use~\eqref{eq:ln} and~\eqref{eq:bound} to prove the desired result, for $\gamma = \frac{\lambda}{2}$.

Let $w_n = u_{h(0)} u_{h(1)} \cdots u_{h(n-1)}$.
Denote the $\limu$-recurrence class of $r$ by $R = \set{q \in Q \mid \limu(r,q) = 1}$.
We first bound the quantity $\pRob{\AA}(r \xrightarrow{w_n} Q \setminus R)$.
Note that for $q \in R$ and $t$ a state, the following holds: if $\limu(q,t) = 1$ then $t \in R$,
\textit{i.e.} $R$ is not left while following transitions consistent with $\limu$.
It follows that the probability to leave $R$ from $r$ while reading $w_n$ is smaller than 
$\sum_{k = 0}^{n-1} \pRob{\AA}(R \xrightarrow{u_{h(k)}} Q \setminus R)$,
which is smaller than $\frac{1}{2}$ by~\eqref{eq:ln}.
Thus $\pRob{\AA}(r \xrightarrow{w_{n-1}} R) \ge \frac{1}{2}$. 
Now, since $r$ is $\limu$-recurrent, and $\limu$ is idempotent, for all $q \in R$ we have $\limu(q,r) = 1$, 
so using~\eqref{eq:bound} we get that $\pRob{\AA}(q \xrightarrow{u_{h(n-1)}} r) \ge \lambda$. 
It follows that $\pRob{\AA}(r \xrightarrow{w_n} r) \geq \frac{\lambda}{2}$, which concludes.
\end{proof}

\begin{lem}\label{lem:nonsimp}
Let $w \in A^\omega$ be an infinite word.
If there exist states $p,s,t$, $(i_n)\nNN$, $(j_n)\nNN$ and $\gamma > 0$ such that:
\begin{enumerate}
	\item for all $n \in \NN$, $\pRob{\AA}(p \xrightarrow{w_{< i_n}} s) \ge \gamma$,
	\item for all $n \in \NN$, $i_n < j_n$ and $\pRob{\AA}(s \xrightarrow{w[i_n,j_n]} t) > 0$,
	\item $\lim_{n \to \infty} \pRob{\AA}(p \xrightarrow{w_{< j_n}} t) = 0$,
\end{enumerate}
then the process induced by $w$ from $p$ is not simple.
\end{lem}

\begin{proof}
Assume towards contradiction that $w$ induces a simple process from $p$ with bound $\lambda$. 
We first argue that for infinitely many $n \in \NN$, 
we have $s \in A_{i_n}$ and $t \in B_{j_n}$. 
Indeed, if this is not the case, then for $n \in \NN$ large enough either $s \notin A_{i_n}$ or $t \notin B_{j_n}$,
so either for infinitely many $n \in \NN$ we have $s \notin A_{i_n}$, or 
for infinitely many $n \in \NN$ we have $t \notin B_{j_n}$.
The first case is contradicted by (1), the second case by (3).

Let $n \in \NN$ such that $s \in A_{i_n}$ and $t \in B_{j_n}$, 
since $\pRob{\AA}(s \xrightarrow{w[i_n,j_n]} t) > 0$,
along a path from $s$ to $t$ there is a transition from the jet $A$ to the jet $B$.
Formally, there exists $k_n$ such that 
$i_n \le k_n < j_n$ and the $k_n$\textsuperscript{th} transition goes from $q_{k_n} \in A_{k_n}$
to $q_{k_n + 1} \in B_{k_n + 1}$. 
This transition is a one-step transition in the automaton $\AA$; denote by $p_{\min}$
the minimal non-zero probabilistic transition in $\AA$, we have
$\pRob{\AA}(q_{k_n} \xrightarrow{w[k_n,k_n+1]} q_{k_n + 1}) \ge p_{\min}$.
Now, consider $\pRob{\AA}(p \xrightarrow{w_{\leq k_n}} q_{k_n + 1})$; 
since $q_{k_n} \in A_{k_n}$, we have $\pRob{\AA}(p \xrightarrow{w_{< k_n}} q_{k_n}) \ge \lambda$,
so $\pRob{\AA}(p \xrightarrow{w_{\leq k_n}} q_{k_n + 1}) \ge \lambda \cdot p_{\min}$.
But for infinitely many $n \in \NN$, $q_{k_n + 1} \in B_{k_n + 1}$, contradicting 
$\lim_{n \to \infty} \pRob{\AA}(p \xrightarrow{w_{< n}} B_n) = 0$.
\end{proof}

Now we are ready to prove Proposition~\ref{prop:non_simplicity_witness_implies_not_simple}.

\begin{proof}
Let $(\limu,\limv,\limw)$ be a non-simplicity witness, 
and $r,t\in Q$ such that $r$ is $\limu \limv^\sharp \limw$-recurrent,
$\limu \limv(r,t) = 1$ and $t$ is $\limv$-transient.
Let $(u_n)\nNN,(v_n)\nNN,(w_n)\nNN$ be sequences of words which reify $\limu,\limv,\limw$ respectively. 

  
Thanks to Proposition~\ref{prop:consistency}, there exists a strictly increasing map $f : \NN \to \NN$ such that 
$(v_{f(n)}^n)\nNN$ reifies $\limv^\sharp$.
Note that since $(v_{f(n)})\nNN$ is a subsequence of $(v_n)\nNN$, it also reifies $\limv$.
Since $\limv$ is idempotent, for all $k \in \NN$,
$(v_{f(n)}^k)\nNN$ reifies $\limv^k$.
By assumption $\limu \limv(r,t) = 1$, so there exists $N_k \in \NN$ 
such that for all $n \geq N_k$, $\pRob{\AA}(r \xrightarrow{u_n v_{f(n)}^k} t) > 0$. 

Let $g(n) = \max(n,N_n)$. Since $g$ is increasing, 
$(u_{g(n)} v_{f(g(n))}^{g(n)} w_{g(n)})\nNN$ is a subsequence of $(u_n v_{f(n)}^n w_n)\nNN$,
so it reifies $\limu \limv^\sharp \limw$ as well.
By definition of the function $g$, we have:
\begin{equation}\label{eq:2}
\text{for all } n \in \NN,\quad \pRob{\AA}(r \xrightarrow{u_{g(n)} v_{f(g(n))}^n} t) > 0.
\end{equation}

We apply Lemma~\ref{lem:rec} to the limit-word $\limu \limv^\sharp \limw$, the state $r$ 
and the sequence of words $(u_{g(n)} v_{f(g(n))}^{g(n)} w_{g(n)})\nNN$,
and obtain $h$ a strictly increasing map and a constant $\gamma > 0$.

Define the new sequence of words $(z_n)\nNN = (u_{g(h(n))} \cdot v_{f(g(h(n)))}^{g(h(n))} \cdot w_{g(h(n))})\nNN$,
and $(x_n)\nNN = (u_{g(h(n))} \cdot v_{f(g(h(n)))}^{h(n)})\nNN$.
We have: 
\begin{equation}\label{eq:gama}
\text{ for all } n \in \NN,\quad \pRob{\AA}(r \xrightarrow{z_0 \cdots z_{n-1}} r) \ge \gamma.
\end{equation}

Let $z = z_0 z_1 \cdots$.
We argue that the conditions of Lemma~\ref{lem:nonsimp} are met:
\begin{enumerate}
	\item for all $n \in \NN$,\quad $\pRob{\AA}(r \xrightarrow{z_0 \cdots z_{n-1}} r) \ge \gamma$,
	\item for all $n \in \NN$, \quad $\pRob{\AA}(r \xrightarrow{x_n} t) > 0$,
	\item $\lim_{n \to \infty} \pRob{\AA}(r \xrightarrow{z_0 \cdots z_{n-1} \cdot x_n} t) = 0$,
\end{enumerate}

The item (1) is~\eqref{eq:gama}, the item (2) follows from~\eqref{eq:2}, so we consider (3). 

First, note that $(x_n)\nNN$ reifies $\limu \limv^\sharp$. 
Indeed, we first argue that $(v_{f(g(n))}^n)\nNN$ reifies $\limv^\sharp$: 
it follows from Proposition~\ref{prop:consistency},
since $f \circ g : \NN \to \NN$ is a strictly increasing map satisfying $f \circ g \ge f$,
and that $(v_{f(n)}^n)\nNN$ reifies $\limv^\sharp$.
Now, $(v_{f(g(h(n)))}^{h(n)})\nNN$ is a subsequence of $(v_{f(g(n))}^n)\nNN$,
so it reifies $\limv^\sharp$ as well, and it follows that $(x_n)\nNN$ reifies $\limu \limv^\sharp$.

Second, let $q \in Q$, since $t$ is $\limv$-transient, $\limu \limv^\sharp(q,t) = 0$. 
Since $(x_n)\nNN$ reifies $\limu \limv^\sharp$, we have 
$\lim_{n \to \infty} \pRob{\AA}(q \xrightarrow{x_n} t) = 0$. 
Now, 
\[
\pRob{\AA}(r \xrightarrow{z_0 \cdots z_{n-1} \cdot x_n} t) 
= \sum_{q \in Q} \pRob{\AA}(r \xrightarrow{z_0 \cdots z_{n-1}} q) 
\cdot \pRob{\AA}(q \xrightarrow{x_n} t),
\]
and for each term in the sum, the second factor converges to zero,
so $\lim_{n \to \infty} \pRob{\AA}(r \xrightarrow{z_0 \cdots z_{n-1} \cdot x_n} t) = 0$.

Thus Lemma~\ref{lem:nonsimp} applies, and $z$ induces a non-simple process from $p$,
so $\AA$ is not simple.
\end{proof}

\subsubsection{The presence of a leak implies the presence of a non-simplicity witness}

Now we show that the presence of a leak witness implies the presence of a non-simplicity witness.

\begin{prop}\label{prop:sw} 
If the extended Markov monoid of a probabilistic automaton contains a leak witness,
then it also contains a non-simplicity witness.
\end{prop}

In the following proof, we will make use of the notion of $\sharp$-height for an extended limit-word.
The $\sharp$-height of a limit-word was defined in Section~\ref{def:sharp_height},
the $\sharp$-height of an extended limit-word $(\limu,\limu_+)$ is the $\sharp$-height of $\limu$.

\begin{proof}
Let $\AA$ be a probabilistic automaton whose extended Markov monoid contains a leak witness.
Consider the subset $\CC$ of extended limit-words $(\limu,\limu_+)$ in the extended Markov monoid such that
there exist states $r,t$ satisfying:
\begin{enumerate}
	\item $(\limu,\limu_+)$ is idempotent,
	\item $r$ is $\limu$-recurrent,
	\item $\limu(r,t) = 0$,
	\item $\limu_+(r,t) = 1$.
\end{enumerate}

Note that $\CC$ is non-empty since every leak witness is in $\CC$.

Consider an element $(\limz,\limz_+)$ in $\CC$ of minimal $\sharp$-height
and let $r,t \in Q$ such that $r$ is $\limz$-recurrent,
$\limz(r,t) = 0$ and $\limz_+(r,t) = 1$.
In particular, we have $\limz \neq \limz_+$, so in any decomposition of $(\limz,\limz_+)$
into concatenation and iteration there must be at least one iteration.
Consequently, $\limz = \limu \limv^\sharp \limw$ for some $\limu,\limv,\limw$,
and $\limu \limv \limw$ has a strictly smaller $\sharp$-height than $\limz$. 

Let $R$ be the $\limz$-recurrence class of $r$, and $T$ the set of $\limv$-transient states.
We argue that the following holds:
\begin{equation}\label{eq:inside}
\text{there exist } r' \in R \text{ and } t' \in T \text{ such that } \limu \limv(r',t') = 1.
\end{equation}
Assume towards contradiction that~\eqref{eq:inside} does not hold,
\textit{i.e.} for all $r' \in R$ and $t' \in T$, we have $\limu \limv(r',t') = 0$,
then we prove that $(\limu \limv \limw, \limz_+)^{|\monoidext |!}$ is in $\CC$,
contradicting the minimality of $(\limz, \limz_+)$ as it has strictly smaller $\sharp$-height.

First observe that for all states $q$ we have $\limu \limv(r',q) = \limu\limv^\sharp(r',q)$,
which implies that for all $r' \in R$ and state $q$, we have:
\begin{equation}\label{eq:same}
\limu \limv \limw(r',q) = \limu \limv^\sharp \limw(r',q) = \limz(r',q).
\end{equation}

We check that $(\limu \limv \limw, \limz_+)^{|\monoidext |!}$ is in $\CC$, with the states $r,t$ as witnesses:
\begin{enumerate}
	\item $(\limu \limv \limw, \limz_+)^{|\monoidext |!}$ is idempotent, this follows from Lemma~\ref{lem:basic_limit_words}.
	\item $r$ is $(\limu \limv \limw)^{|\monoidext |!}$-recurrent. 
	Indeed, let $q \in Q$ such that $(\limu \limv \limw)^{|\monoidext |!}(r,q) = 1$.
	It follows that  $(\limu \limv^\sharp \limw)^{|\monoidext |!}(r,q) = 1$, but
	$(\limu \limv^\sharp \limw)^{|\monoidext |!} = \limz^{|\monoidext |!} = \limz$, 
	so $\limz(r,q) = 1$. 
	Since $r$ is $\limz$-recurrent, we have $\limz(q,r) = 1$ and $q \in R$, 
	so~\eqref{eq:same} implies that $\limz(q,r) = \limu \limv \limw(q,r)$, thus $\limu \limv \limw(q,r) = 1$.
	Also, $\limz(r,r) = \limu \limv \limw(r,r)$, so $\limu \limv \limw(r,r) = 1$,
	and altogether $(\limu \limv \limw)^{|\monoidext |!}(q,r) = 1$.
	\item $(\limu \limv \limw)^{|\monoidext |!}(r,t) = 0$.
	Indeed, assume towards contradiction that $(\limu \limv \limw)^{|\monoidext |!}(r,t) = 1$.
	Then there exist $q_0,q_1,\ldots,q_{|\monoidext |!}$ such that
	$q_0 = r$, for $i \in \set{0,\ldots,|\monoidext |!-1}$ we have $\limu \limv \limw(q_i,q_{i+1}) = 1$
	and $q_{|\monoidext |!} = t$.
	We prove by induction on $i \in \set{0,\ldots,|\monoidext |!-1}$ that $q_i$ is in $R$
	and that $\limz(q_i,q_{i+1}) = 1$.
	Assume $q_i$ is in $R$, then 
	by~\eqref{eq:same}, $\limu \limv \limw(q_i,q_{i+1}) = \limz(q_i,q_{i+1})$,
	so $\limz(q_i,q_{i+1}) = 1$. But $q_i$ is in $R$, which is a recurrence class for $\limz$,
	so $q_{i+1}$ is in $R$ as well, concluding the induction.
	Thus, we have $\limz^{|\monoidext |!}(r,t) = 1$, and since $\limz$ is idempotent $\limz(r,t) = 1$, a contradiction.
 	\item $\limz_+^{|\monoidext |!}(r,t) = 1$. It follows from the fact the $\limz_+$ is idempotent,
 	and that $\limz_+(r,t) = 1$.
\end{enumerate}

We reached a contradiction, since $(\limu \limv \limw, \limz_+)^{|\monoidext |!}$ has a strictly smaller $\sharp$-height than 
$(\limz, \limz_+)$.
It follows from~\eqref{eq:inside} that $(\limu,\limv,\limw)$ is a non-simplicity witness for the states $r'$ and $t'$,
concluding the proof.
\end{proof}

The proof of Theorem~\ref{theo:simple_inclusion_leaktight} is now a simple combination of 
Proposition~\ref{prop:non_simplicity_witness_implies_not_simple}
and Proposition~\ref{prop:sw}.
The inclusion is strict: figure~\ref{fig:3} provides an example of
a leaktight automaton which is not simple.

\section{Probabilistic \texorpdfstring{$\omega$}{omega}-automata}
\label{sec:infinite}
In this section, we relate the value $1$ problem for probabilistic automata over finite words
and over infinite words, as introduced and studied in~\cite{BBG12}.
We state a general theorem, showing the equivalence of the value $1$ problem for automata over finite words
and the value $1$ problem for automata over infinite words with the parity condition.
This theorem allows to extend the decidability results from finite words to infinite ones.

\vskip1em
For the definitions of probabilistic automata over infinite words, we refer to~\cite{BBG12}.
For the sake of readability, we introduce two notations in this section.
First, we denote by $\pRob{\AA,w}(E)$ the probability of the measurable event $E$
when reading the infinite word $w$ on $\AA$.
Second, we denote by $\pRob{\AA}^{\delta}(u)$
the probability that the finite word $u$ is accepted by $\AA$ with $\delta$ as initial distribution.
Note that in our definition, probabilistic automata have a unique initial state;
however here we need to deal with automata having general initial distributions, 
so we sometimes consider the more general tuples 
$\AA = (Q, \delta, \Delta, F)$ where $\delta$ is the initial probability distribution.

\begin{thm}\label{theo:infinite_words}
Let $\AA = (Q,\delta_0,\Delta,c)$ be a probabilistic parity automaton
where $c : Q \rightarrow \NN$ is a priority function.
Consider the deterministic transducer $\MM$ over $\AA$, which 
keeps track of the minimal priority seen:
$\MM = (c(Q), Q^\AA, \Delta_\MM)$ where $\Delta_\MM(q,c) = \min(c,c(q))$.

The automaton $\AA$ over infinite words has value $1$ if and only if there exists $R \subseteq Q$, such that
the two following probabilistic automata over finite words have value $1$:
\begin{itemize}
	\item The automaton $\AA(R)$ with $R$ as set of final states;
	\item The automaton $\AA \otimes \MM$ with the uniform distribution over 
	$R_c = \set{(q,c(q)) \mid q \in R}$ as initial distribution
	and $\set{(q,e) \mid q \in R \textrm{ and } e \textrm{ even}}$ as set of final states.
\end{itemize} 
\end{thm}

\noindent Before giving the proof, we need two lemma.

\begin{lem}\label{lem:technical}
Let $\AA = (Q,\delta,\Delta,F)$ be a probabilistic automaton,
$u$ a word and $\varepsilon > 0$.
Let $\mu = \min \set{\delta(q) \mid q \in \supp(\delta)}$.
If $\pRob{\AA}^{\delta}(u) \ge 1 - \varepsilon \cdot \mu$, then for all
$q \in \supp(\delta)$, $\pRob{\AA}^q(u) \ge 1 - \varepsilon$.
\end{lem}

\begin{proof}
$$\pRob{\AA}^{\delta}(u) = \sum_{q \in \supp(\delta)} \delta(q) \cdot \pRob{\AA}^q(u) \ \ge \ 1 - \varepsilon \cdot \mu.$$
Let $q \in \supp(\delta)$, since the probabilities are bounded by $1$:
$$\delta(q) \cdot \pRob{\AA}^q(u)\ + \ \underbrace{\sum_{p \in \supp(\delta),\ p \neq q} \delta(p)}_{1 - \delta(q)} \ \ge \ 
1 - \varepsilon \cdot \mu.$$
Hence:
$$\pRob{\AA}^q(u)\ \ge \ 1 - \varepsilon \cdot \frac{\mu}{\delta(q)} \ \ge \ 1 - \varepsilon.$$
\end{proof}

\begin{cor}\label{cor:value_1_independence}
Let $\AA = (Q,\delta,\Delta,F)$ be a probabilistic automaton.
If $\AA$ has value $1$, then:
\begin{itemize}
	\item For all distributions $\delta'$ such that $\supp(\delta') \subseteq \supp(\delta)$,
it has value $1$ with $\delta'$ as initial distribution.
	\item For all distributions $\delta'$ such that $\sum_{q \in \supp(\delta)} \delta'(q) \ge \theta$,
it has value at least $\theta$ with $\delta'$ as initial distribution.
\end{itemize}
\end{cor}

\begin{proof}
Assume $\AA$ has value $1$. Let $\varepsilon > 0$, then there exists a word $u$ such that
$\pRob{\AA}(u) \geq 1 - \varepsilon \cdot \mu$.
Thanks to Lemma~\ref{lem:technical}, this implies that for all $q \in \supp(\delta)$,
we have $\pRob{\AA}^q(u) \geq 1 - \varepsilon$.

Now for $\delta'$:
$$\pRob{\AA}^{\delta'}(u)\ = \ \sum_{q \in Q} \delta'(q) \cdot \pRob{\AA}^q(u) \ \ge \ 
\sum_{q \in \supp(\delta)} \delta'(q) \cdot \pRob{\AA}^q(u) \ \ge \ 
\left(\sum_{q \in \supp(\delta)} \delta'(q)\right) \cdot (1 - \varepsilon).$$
For the first item, note that if $\supp(\delta') \subseteq \supp(\delta)$, 
then $\sum_{q \in \supp(\delta)} \delta'(q) = 1$,
so $\pRob{\AA}^{\delta'}(u) \geq 1 - \varepsilon$,
so $\AA$ has value $1$ with $\delta'$ as initial distribution.

The second item follows from the last inequality, implying $\pRob{\AA}^{\delta'}(u) \geq \theta \cdot (1 - \varepsilon)$,
so $\AA$ has value at least $\theta$ with $\delta'$ as initial distribution.
\end{proof}

\begin{lem}\label{lem:seq}
For all $\varepsilon > 0$, there exists a sequence $(\varepsilon_k)_{k \ge 0}$ satisfying:
\begin{enumerate}
	\item For all $k \ge 0$, we have $0 < \varepsilon_k < 1$;
	\item $\prod_{k \ge 0} \varepsilon_k \geq 1 - \varepsilon$;
	\item For all $k \ge 0$, we have $\prod_{p \le k} \varepsilon_p > \varepsilon_{k+1}$.
\end{enumerate}
\end{lem}

\begin{proof}
Define $\nu_k = 1 - \frac{1}{2^{2^k}}$, which clearly satisfies 1.
It also satisfies 3.:
$$\begin{array}{lll}
\prod_{p \le k} \nu_p & = & \prod_{p \le k} \frac{2^{2^p} - 1}{2^{2^p}} \\[1ex]
& = & \frac{\prod_{p \le k} 2^{2^p} - 1}{2^{2^{k+1} - 1}}
\end{array}$$
The inequality $\prod_{p \le k} \nu_p > \nu_{k+1}$ is equivalent to:
$$2 \cdot \left(\prod_{p \le k} (2^{2^p} - 1)\right) < 2^{2^{k+1}} - 1,$$
which is easily proved by induction, since $(2^{2^{k+1}}  - 1)^2 < 2^{2^{k+2}}  - 1$.

The infinite product $\prod_{k \ge 0} \nu_k$ has a value $\nu \approx 0.350184$, in particular $0 < \nu < 1$.
Let $\varepsilon > 0$, then for $\lambda = \frac{\ln(1-\varepsilon)}{\ln(\nu)}$,
which satisfies $\lambda > 0$, the sequence $(\nu_k^\lambda)_{k \in \NN}$ satisfies the three conditions.
\end{proof}

We are now fully equipped for the proof of Theorem~\ref{theo:infinite_words}
\vskip1em

We begin with the left-to-right direction.
Let $(w_n)_{n \in \NN}$ be a sequence of infinite words such that $\pRob{\AA,w_n}(\parc)$ converges to $1$.
For each $n$, denote by $(\delta_k^n)_{k \in \NN}$ the sequence of distributions assumed by $\AA$
when reading $w_n$.
Since there are finitely many possible supports, there exists
a subsequence where all distributions have the same support.
Since $[0,1]^Q$ is a compact space, we can extract from this subsequence another one,
denoted by $(\delta_{\phi_n(k)}^n)_{k \in \NN}$,
which converges to a distribution denoted by $\delta^n$.
By the same compactness argument, at the expense of considering a subsequence of $(w_n)_{n \in \NN}$,
we can assume that $(\delta^n)_{n \in \NN}$ converges to a distribution $\delta$,
whose support is denoted by $R$.

We now prove that $\AA(R)$ as well as $\AA \otimes \MM$ have value~$1$.

\medskip\textbf{The automaton $\AA(R)$.}
Let $\varepsilon > 0$.
Since $(\delta^n)_{n \in \NN}$ converges to $\delta$,
for $n$ large enough, we have $||\delta^n - \delta||_\infty \le \frac{\varepsilon}{2 \cdot |R|}$.
We fix $n$ large enough;
since $(\delta_{\phi_n(k)}^n)_{k \in \NN}$ converges to $\delta^n$,
for $k$ large enough, we have $||\delta_{\phi_n(k)}^n - \delta^n||_\infty \le \frac{\varepsilon}{2 \cdot |R|}$.
We fix such a large $k$, altogether this implies
$||\delta_{\phi_n(k)}^n - \delta||_\infty \le \frac{\varepsilon}{|R|}$.

Now consider $u$ the prefix of $w_n$ until position $\phi_n(k)$.
We have:
$$\pRob{\AA(R)}(u) \ =\ \sum_{q \in R} \delta_{\phi_n(k)}^n(q) \ \ge \
\sum_{q \in R} \left(\delta(q) - \frac{\varepsilon}{|R|}\right) \ = \ 1 - \varepsilon.$$
It follows that $\AA(R)$ has value $1$.

\medskip\textbf{The automaton $\AA \otimes \MM$.}
Let $\nu = \min \set{\delta(q) \mid q \in \supp(\delta)}$, clearly $0 < \nu < 1$.
Since $(\delta^n)_{n \in \NN}$ converges to $\delta$,
we can assume that for all $n \ge 0$, $||\delta^n - \delta||_\infty \le \frac{\nu}{4}$,
considering the sequence $(w_n)_{n \in \NN}$ from some index on.
For each $n$, the sequence $(\delta_{\phi_n(k)}^n)_{k \in \NN}$ converges to $\delta^n$,
so again by considering the sequence from some index on,
we can assume that for all $k \ge 0$, $||\delta_{\phi_n(k)}^n - \delta^n||_\infty \le \frac{\nu}{4}$,
hence $||\delta_{\phi_n(k)}^n - \delta||_\infty \le \frac{\nu}{2}$.
As a result, for all $q \in R$, we have $\delta_{\phi_n(k)}^n(q) \ge \frac{\nu}{2} > 0$,
so $R = \supp(\delta) \subseteq \supp(\delta^n) \subseteq \supp(\delta_{\phi_n(k)}^n)$.

\vskip1em
Let $\varepsilon > 0$, and $n$ such that $\pRob{\AA,w_n}(\parc) \ge 1 - \varepsilon \cdot \frac{\nu}{4}$.

Let $\buchi(d)$ be the set of runs where the color $d$ is reached infinitely often
and by $\cobuchi(<d)$ the set of runs where colors less than $d$ are reached only finitely often.
Observe that:
$$\parc\ =\ \biguplus_{d \textrm{ even color}} \buchi(d) \cap \cobuchi(<d).$$ 

Let $\safe(<d,k)$ be the set of runs where colors less than $d$ are not reached anymore after the position $\phi_n(k)$.
Observe that:
$$\buchi(d) \cap \cobuchi(<d)\ \subseteq\ \buchi(d) \cap \bigcup_{k \ge 0} \safe(<d,k),$$ 
so:
$$\pRob{\AA,w_n}\left(\buchi(d) \cap \cobuchi(<d)\right)\ \leq\ 
\pRob{\AA,w_n}\left(\buchi(d) \cap \bigcup_{k \ge 0} \safe(<d,k)\right).$$
Since the sequence $\buchi(d) \cap \left(\safe(<d,k)\right)_{k \ge 0}$ is increasing with respect to set inclusion, 
for some $k$ we have:
$$\pRob{\AA,w_n}(\buchi(d) \cap \safe(<d,k))\ \ge \ 
\pRob{\AA,w_n}(\buchi(d) \cap \cobuchi(<d)) - \varepsilon \cdot \frac{\nu}{4}.$$
Note that here, $k$ depends on $d$; however, since if it holds for $k$ then it holds for any bigger $k$
and that there are finitely many even colors $d$, we can assume that $k$ is uniform
over all even colors $d$.

Now, let $\reach(d,k,k')$ be the set of runs where the color $d$ is reached between the positions $\phi_n(k)$ and $\phi_n(k')$.
Observe that:
$$\buchi(d) \cap \safe(<d,k)\ \subseteq\ \bigcup_{k' > k} \reach(d,k,k') \cap \safe(<d,k),$$ 
so:
$$\pRob{\AA,w_n}(\buchi(d) \cap \safe(<d,k))\ \leq\ 
\pRob{\AA,w_n}\left(\bigcup_{k' > k} \reach(d,k,k') \cap \safe(<d,k)\right).$$
Since the sequence $\left(\reach(d,k,k') \cap \safe(<d,k)\right)_{k' > k}$ is increasing 
with respect to set inclusion, for some $k'$ we have:
$$\pRob{\AA,w_n}(\reach(d,k,k') \cap \safe(<d,k))\ \ge \ 
\pRob{\AA,w_n}\left(\buchi(d) \cap \safe(<d,k)\right) - \varepsilon \cdot \frac{\nu}{4}.$$
Here again, $k'$ depends on $d$; however with the same reasoning as above, 
we can assume that $k'$ is uniform over all even colors $d$.

Altogether, this implies:
$$\pRob{\AA,w_n}(\reach(d,k,k') \cap \safe(<d,k))\ \ge \ 
\pRob{\AA,w_n}(\buchi(d) \cap \cobuchi(<d)) - \varepsilon \cdot \frac{\nu}{2},$$
and summing these equalities for each even color $d$: 
$$\begin{array}{lcl}
\pRob{\AA,w_n}(\parity(c,k,k')) & \ge & \sum_{d \textrm{ even color}} \pRob{\AA,w_n}(\reach(d,k,k') \cap \safe(<d,k))\\[.6em]
& \ge & \pRob{\AA,w_n}(\parc) - \varepsilon \cdot \frac{\nu}{2},
\end{array}$$
where $\parity(c,k,k')$ is the set of runs where
the minimal color seen between the positions $\phi_n(k)$ and $\phi_n(k')$ is even.

Let $v$ be the infix of $w_n$ between the positions $\phi_n(k)$ and $\phi_n(k')$.
The distribution $\delta_{\phi_n(k)}^n$ is over the set of states $Q$;
we embed it as $\widehat{\delta}$ over the set of states $Q \times c(Q)$,
setting $\widehat{\delta}(q,c(q)) = \delta_{\phi_n(k)}^n(q)$ and $0$ otherwise.
By construction of $v$, we have $\pRob{\AA \otimes \MM}^{\widehat{\delta}}(v) = 
\pRob{\AA,w_n}(\parity(c,k,k')) \ge 1 - \varepsilon \cdot \frac{\nu}{2}$.
However, the initial distribution of $\AA \otimes \MM$ is uniform over $R_c$.
We apply Lemma~\ref{lem:technical}; by construction $\mu \ge \frac{\nu}{2}$,
so $\pRob{\AA \otimes \MM}^{\widehat{\delta}}(v) \ge 1 - \varepsilon \cdot \mu$.
It follows that for all $q \in \supp(\delta_{\phi_n(0)}^n)$, $\pRob{\AA \otimes \MM}^{q,c(q)}(v) \ge 1 - \varepsilon$.
Since $R \subseteq \supp(\delta_{\phi_n(k)}^n)$, we have $\pRob{\AA \otimes \MM}(v) \ge 1 - \varepsilon$,
hence $\AA \otimes \MM$ has value $1$.

\vskip2em
We now prove the right-to-left direction.
Let $R$ such that both $\AA(R)$ and $\AA \otimes \MM$ have value $1$.
Let $\varepsilon > 0$, consider the sequence $(\varepsilon_k)_{k \ge 0}$ given by Lemma~\ref{lem:seq}.

Since $\AA(R)$ has value $1$, there exists $u$ such that $\pRob{\AA(R)}(u) \ge \varepsilon_0$.
We show the existence of a sequence $w_1,w_2,\ldots$ of words such that for all $k \ge 0$:
$$\pRob{\AA(R)}(u \cdot w_1 \cdot w_2 \cdot \ldots \cdot w_k) \ge \prod_{p \le k} \varepsilon_p
\qquad \mathrm{ and } \qquad \pRob{\AA \otimes \MM}^{\delta'_k}(w_{k+1}) \ge \varepsilon_{k+1},$$
where $\delta_k$ is the distribution obtained by reading $u \cdot w_1 \cdot w_2 \cdot \ldots \cdot w_k$
on $\AA$; the distribution $\delta_k$ is over the set of states $Q$,
we embed it as $\delta'_k$ over the set of states $Q \times c(Q)$,
setting $\delta'_k(q,c(q)) = \delta_k(q)$ and $0$ otherwise.

We proceed inductively; assume $w_1,\ldots,w_k$ have been chosen.
Note that by construction, $\sum_{q \in R} \delta_k(q) \ge \prod_{p \le k} \varepsilon_p$.
However, the support of $\delta_k$ is not included in $R$, so the first item of Corollary~\ref{cor:value_1_independence}
does not apply.
Then $\sum_{q \in R_c} \delta'_k(q) \ge \prod_{p \le k} \varepsilon_p$.

Since $\AA \otimes \MM$ has value $1$, thanks to the second item of Corollary~\ref{cor:value_1_independence},
it has value at least $\prod_{p \le k} \varepsilon_p$ for $\delta'_k$ as initial distribution.
Together with 3., this implies that 
there exists $w_{k+1}$ such that $\pRob{\AA \otimes \MM}^{\delta'_k}(w_{k+1}) \ge \varepsilon_{k+1}$.

We have:
$$\begin{array}{lll}
\pRob{\AA(R)}(u \cdot w_1 \cdot w_2 \cdot \ldots \cdot w_{k+1}) & \ge & 
\pRob{\AA(R)}(u \cdot w_1 \cdot w_2 \cdot \ldots \cdot w_k) \cdot
\pRob{\AA(R)}^{\delta_k}(w_{k+1}) \\
& \ge & \pRob{\AA(R)}(u \cdot w_1 \cdot w_2 \cdot \ldots \cdot w_k) \cdot
\pRob{\AA \otimes \MM}^{\delta'_k}(w_{k+1}) \\
& \ge & \left(\prod_{p \le k} \varepsilon_p\right) \ \cdot \ \varepsilon_{k+1}\\
& = & \prod_{p \le k+1} \varepsilon_p,\\
\end{array}$$
which concludes the inductive construction.

Let $w = u \cdot w_1 \cdot w_2 \cdot \ldots$. We evaluate $\pRob{\AA,w}(\parc)$:
$$\pRob{\AA,w}(\parc) \ \ge \ \pRob{\AA(R)}(u) \cdot \prod_{k \ge 0} \pRob{\AA \otimes \MM}^{\delta'_k}(w_{k+1})
\ \ge \ \varepsilon_0 \cdot \prod_{k \ge 1} \varepsilon_k
\ = \ \prod_{k \ge 0} \varepsilon_k \ \ge \
1 - \varepsilon.$$
It follows that $\AA$ has value $1$ as a probabilistic parity automaton,
which concludes the proof of Theorem~\ref{theo:infinite_words}.

\begin{cor}
The value $1$ problem for leaktight automata over infinite words with the parity condition is decidable
and $\PSPACE$-complete.
\end{cor}

Indeed, observe that if $\AA$ is leaktight,
then by Proposition~\ref{prop:deterministic_transducer},
so is $\AA \otimes \MM$.

\section*{Acknowledgment}
We thank Thomas Colcombet for having pointed us to the work of Leung and Simon,
and the anonymous reviewers for their constructive comments.

\section*{Conclusion}
We introduced a subclass of probabilistic automata, called leaktight automata,
for which we proved that the value $1$ problem is $\PSPACE$-complete.
This subclass generalizes all subclasses of probabilistic automata
whose value $1$ problem is known to be decidable.

A challenging perspective is now to find subclasses of partially observable Markov decision processes where the value 1 problem is
decidable (some preliminary results were given in~\cite{GO14}),
and to extend our results to the setting of partially observable stochastic games, which is even more challenging.

\bibliographystyle{alpha}
\bibliography{bib}

\end{document}